\documentclass{article}

\usepackage{amsmath,amsfonts}

\title{%
  \MakeUppercase{Graph Skeletonization of High-Dimensional Point Cloud Data via Topological Method}%
  \thanks{This work is partially supported by National Science Foundation (NSF) under grants CCF-2051197, RI-1815697, and OAC-2107076, and by National Institutes of Health (NIH) under grant RF1MH125317.}
}

\author{%
  Lucas Magee,%
  \thanks{Department of Computer Science and Engineering, University of California, San Diego, 
          lmagee@ucsd.edu}\,
  Yusu Wang%
  \thanks{Hal{\i}c{\i}o\u{g}lu Data Science Institute, University of California, San Diego,
          yusuwang@ucsd.edu}\
}

\usepackage[utf8]{inputenc} % allow utf-8 input
\usepackage[T1]{fontenc}    % use 8-bit T1 fonts
\usepackage{hyperref}       % hyperlinks
\usepackage{url}            % simple URL typesetting
\usepackage{booktabs}       % professional-quality tables
\usepackage{amsfonts}       % blackboard math symbols
\usepackage{nicefrac}       % compact symbols for 1/2, etc.
\usepackage{microtype}      % microtypography
\usepackage{xcolor}         % colors
\usepackage{tikz-cd}
\usepackage{tikz}
\usepackage{adjustbox}
\usepackage[algo2e,ruled,vlined]{algorithm2e}
\usepackage[normalem]{ulem}
\usepackage{multirow}
\usepackage{algpseudocode}
\usepackage{amsmath,amsthm}
\usepackage{float}

\newcommand{\eps}  {{\varepsilon}}
\newcommand{\reals}     {{\mathbb{R}}}

\newtheorem{theorem}{Theorem}[section]
\newtheorem{lemma}[theorem]{Lemma}

\newtheorem{claim}[theorem]{Claim}

\newtheorem{definition}[theorem]{Definition}

\newcommand{\pers}      {{\mathrm{pers}}}

\newcommand{\dgm}       {{\mathrm{dgm}}}

\newcommand{\Fcal}      {{\mathcal{F}}}
\newcommand{\Tcal}      {{\mathcal{T}}}

\newcommand{\pp}        {{\mathsf{p}}}

\newcommand{\rips}      {{\mathrm{rips}}}
 
\newcommand{\myo}       {{\mathsf{ind}}}
\newcommand{\myprec}    {{\prec}}
\newcommand{\mypreceq}  {{\preceq}}
\newcommand{\myC}       {{\mathsf{C}}}

\newcommand{\ZZ}        {{\mathbb{Z}}}
\newcommand{\homo}      {{\mathsf{H}}}
\newcommand{\myparagraph}[1]  {\vspace*{0.02in}\noindent{\bf #1}~}
\newcommand{\argmax}    {{\mathrm{argmax}}}
\newcommand{\lexoptper}  {{lex-opt persistent}}
\newcommand{\Ghat}      {{\widehat{G}}}
\newcommand{\Khat}      {{\widehat{K}}}
\newcommand{\mye}       {{\mathrm{e}}}
\newcommand{\Tcalhat}  {{\widehat{\Tcal}}}
\newcommand{\bb}        {{\mathsf{b}}}
\newcommand{\dd}        {{\mathsf{d}}}
\newcommand{\DMfirst}   {{\sf firstDM-graph}}
\newcommand{\oldDM}         {{\sf DM-graph}}
\newcommand{\newDM}     {{\sf extDM-graph}}
\newcommand{\DMPCD}     {{\sf DM-PCD}}
\newcommand{\lst}       {{\mathsf{lowSt}}}
\newcommand{\CC}        {{\mathsf{C}}}
\newcommand{\Bsf}       {{\mathsf{B}}}
\newcommand{\Zsf}       {{\mathsf{Z}}}
\newcommand{\mypath}    {{path_{\Tcal_\delta}}}
\newcommand{\mybaseline} {{\sf baseline}}
\newcommand{\wR}        {{\sf wRip}}
\newcommand{\mydomain}  {{\Omega}}
\newcommand{\xx}        {{\mathsf{x}}}
\newcommand{\myorder}   {{\Pi}}
\newcommand{\newrho}    {{\widehat{\rho}}}
\newcommand{\DTMrho}    {{\rho_{w}}}
\newcommand{\myReeb}    {{\sf ReebRecon}}
\newcommand{\myMapper} {{\sf Mapper}}

\begin{document}
\maketitle

\begin{abstract}
Geometric graphs form an important family of hidden structures behind data. In this paper, we develop an efficient and robust algorithm to infer a graph skeleton of a high-dimensional point cloud dataset (PCD). 
Previously, there has been much work to recover a hidden graph from a low-dimensional density field, or from a relatively clean high-dimensional PCD. 
Our proposed approach builds upon the recent line of work on using a persistence-guided discrete Morse (DM) theory based approach to reconstruct a geometric graph from a density field defined over 
a low-dimensional triangulation. 
In particular, we first give a very simple generalization of this DM-based algorithm from a density-function perspective to a general filtration perspective. On the theoretical front, we show that the output of the generalized algorithm contains a so-called lexicographic-optimal persistent cycle basis w.r.t the input filtration, justifying that the output is indeed meaningful. On the algorithmic front, the generalization allows us to combine sparsified weighted Rips filtration to develop a new graph reconstruction algorithm for noisy point cloud data. 
The new algorithm is robust to background noise and non-uniform distribution of input points, and we provide various experimental results to show its effectiveness.
\end{abstract}

\section{Introduction}

Modern complex data, or the space where data is sampled from, often has a simpler underlying structure. A key step in modern data analysis is to model and extract such hidden structures. 
A particularly interesting type of non-linear structure is a (geometric) graph skeleton, 
which can be thought of as a 1-D \emph{singular manifold}, consisting of pieces of 1-manifolds (curves) glued together. 
Graph structures are common in practice, such as river networks and dark matter filament structures in cosmology. Graphs can also be natural models for the evolution of trends behind data (e.g, the evolution of topics in twitter data).  

While there has been beautiful work on manifold learning \cite{Roweis2000,Tenenbaum2000,BN03,donoho2003hel}, recovering singular manifolds is more challenging \cite{BQWZ12}. Nevertheless, recovering a hidden graph skeleton (singular 1-manifolds)
from data has attracted much attention; e.g, in \cite{Hastie84,Kegl02,Ozertem11}.
In general, one of the main challenges involved is to identify graph nodes and connections among them.  Local information is often used to make inference or decisions, making it hard to handle noise, non-uniform sampling and gaps in data.  
To this end, topological methods become useful, as they offer ways to capture the global structure behind data and thus can be robust in detecting junction nodes and their global connectivity. 
Indeed, there are several algorithms that extract a graph skeleton behind point cloud data (PCD) based on topological ideas; e.g.  \cite{Aanjaneya11,GSBW11,CHS15,LRW14}. Unfortunately, while such approaches work well when the input points are sampled within a tubular neighborhood of the hidden graph (called tubular or Hausdorff noise), they do not effectively handle more general noise, such as outliers and background noise. 
The locally-defined principal curve approach  \cite{Ozertem11} can handle noisy data with non-tubular noise via a ridge-finding strategy using a constraint mean-shift-like procedure. However, the procedure only moves points closer to a graph skeleton  without outputting an actual graph. 

Recently, there has been a line of work using a persistence-guided discrete Morse theory based approach to reconstruct a graph (or even a 2D) skeleton from \emph{density field} \cite{DRS15,GDN07,RWS11,2011MNRAS,WWL15}. 
In particular, assume that the input is a density field defined on a discretized domain.
Such methods use the discrete Morse (DM) theory to compute the so-called stable 1-manifolds to capture the mountain ridges of the density field and returns these mountain ridges as the extracted graph skeleton; see Figure \ref{fig:morse_neuron} for a 2D example. Persistent homology is used to simplify the resulting stable 1-manifolds. 
The algorithm based on this idea has been significantly simplified in \cite{DWW18} together with theoretical analysis. 
The resulting method (which we will refer to as \oldDM{}) can recover a hidden graph from noisy and non-homogeneous density fields, and has already been applied to several applications in 2D/3D \cite{BMWL20,DWW19,DWW17}. These graphs
have also been used as input for Graph Neural Networks (GNNs) to generate effective predictive models for rock data \cite{cai2021equivariant}.
However, this method currently assumes that one has \emph{a discretization of the ambient space where data is embedded in}, which becomes prohibitively expensive for high dimensional data, 
and also cannot be directly applied to metric data that is not embedded. 

\paragraph{New work.} 
We consider the general setting where the input is just a set of points $P$ embedded in a metric space, say the Euclidean space $\reals^d$, 
or with pairwise distances (or correlations) given. 
The previous \oldDM{} does not work in this setting, and as we will explain later, the straightforward extension is not effective for high-dimensional PCDs. 
In this paper, we extend the idea behind the discrete-Morse based approach beyond density field, and combine it with the so-called sparsified weighted Rips filtration of \cite{BCOS15} to develop an effective and efficient algorithm to infer graph skeletons of high-dimensional PCDs. 

More specifically, in Section \ref{sec:generalalg}, we view the DM-graph reconstruction method from a filtration perspective instead of a density perspective, and thus generalize the \oldDM{} algorithm to work with an arbitrary filtration (which intuitively is a sequence of growing spaces spanned by our input points in our setting).  
We then prove (Theorem \ref{thm:lexOPHC}) that the output of the generalized method contains a so-called \emph{lex-optimal persistent cycle basis} of the given filtration, thereby showing that the output captures meaningful information w.r.t. the filtration. This result is of independent interest. 

We next show how this simple change of view can help us reconstruct the graph skeleton of a set of points $P$ more efficiently and effectively. In particular, the filtration perspective now allows us to combine the DM-based graph reconstruction algorithm with a sparsified weighted Rips filtration scheme proposed by \cite{BCOS15}, which both improves the quality of the reconstruction and significantly reduces the time complexity. 
This new graph reconstruction algorithm for PCDs, called \DMPCD{}, is our second main contribution and presented in Section \ref{sec:PCD}.    

Finally, we show experimental results on a range of datasets, and compare with previous methods to demonstrate the effectiveness of our new \DMPCD{} algorithm. More results are shown in the Appendix.

\section{Preliminaries}
\label{sec:preliminaries}

We now briefly introduce some notions needed to describe the idea behind the \oldDM{} algorithm of \cite{WWL15, DWW18}. In this paper we will use the {\bf simplicial setting}, where the space of interest is modeled by a \emph{simplicial complex $K$}, consisting of basic building blocks called \emph{simplices}. Intuitively, a geometric $d$-simplex is the convex combination of $d+1$ affinely independent vertices: a $0$-, $1$-, $2$-, or $3$-simplex is just a vertex, an edge, a triangle, or a tetrahedron, respectively. Ignoring the geometry, \emph{an abstract $d$-simplex} $\sigma = (v_0, \ldots, v_d)$ is simply a set of $d+1$ vertices. Any subset $\tau$ of the vertices of a $d$-simplex $\sigma$ is \emph{a face} of $\sigma$, and $\tau$ is called a \emph{facet} of $\sigma$ if its dimension is $d-1$. 
A simplicial complex $K$ is a collection of simplices with the property that if a simplex $\sigma$ is in $K$, then any of its face must be in $K$ as well. 
Given a simplicial complex $K$, its \emph{$q$-skeleton $K^q$} consists of all simplices in $K$ of dimension at most $q$.  

\subsection{Persistent Homology}
\label{subsec:PHintro}

Instead of introducing persistent homology in its full general form, below we focus on the simplicial complex setting. See e.g., \cite{EH10,CD18} for more detailed exposition. 

\paragraph{Boundaries, cycles, homology groups.}
Given a simplicial complex $K$, let $K^q$ denote the set of $q$-simplices of $K$. Under $\ZZ_2$ field coefficient (which we use throughout this paper), 
a $q$-chain $C = \sum_{\sigma \in K^q} c_\sigma \sigma$ where $c_\sigma \in \{0, 1\}$; equivalently $C$ is a subset of $K^q$ (those with $c_\sigma = 1$). 
The set of $q$-chains together with addition operation gives rise to the so-called \emph{$q$-th chain group $\CC_q(K)$}. 
Given any $q$ simplex $\sigma$, its boundary $\partial_q \sigma$ consists of all of its faces of dimension $q$-$1$. 
This in turn gives a linear map, called the $q$-th boundary map $\partial_q: \CC_q(K) \to \CC_{q-1}(K)$, where $\partial_q C = \sum_{\sigma\in K^1} c_\sigma \partial_q(\sigma)$ for any $q$-chain $C = \sum_{\sigma \in K^q} c_\sigma \sigma$. 
A $q$-chain $C$ is a $q$-cycle if its boundary $\partial_q C = 0$. The collection of all $q$-cycles form 
the \emph{$q$-th cycle group $\Zsf_q$}; that is, $\Zsf_q = \mathrm{kernel}~ \partial_q$. 
A $q$-chain $C$ is a $q$-boundary if it is the image of some ($q$+$1$)-chain $C'$; i.e., $C = \partial_{q+1} C'$. 
The collection of $q$-boundaries form the \emph{$q$-th boundary group $\Bsf_q$}; that is, $\Bsf_q= \mathrm{image}~ \partial_{q+1}$. 
By the fundamental property of boundary map, i.e, $\partial_{q} \circ \partial_{q+1} = 0$, it follows that $\Bsf_q$ is a subgroup of $\Zsf_q$. 
The \emph{$q$-th homology group $\homo_q$} is defined as $\homo_q = \Zsf_q / \Bsf_q$. In particular, given any $q$-cycle $C$, its homology class $[C]$ is the equivalent class of all $q$-cycles in $q + \Bsf_{q}(K)$; and two $q$-cycles $C_1, C_2$ are homologous if $[C_1] = [C_2]$, implying that $C_1+C_2$ is a boundary (i.e, $C_1+C_2  \in \Bsf_q(K)$). The \emph{$q$-th homology classes} intuitively capture 
%(under certain assumptions) 
$q$-dimensional ``holes" in $K$; i.e., connected components ($0$D), loops ($1$D), closed surfaces that are not ``filled" ($2$D) and their higher dimensional analogs. 
The $q$th homology group is the  vector space spanned by such topological features, and its rank, called the \emph{$q$-th Betti number} $\beta_q(K)$, gives the number of independent topological "holes". 

\paragraph{Filtration, persistent modules.} 
Suppose we have a finite sequence of simplicial complexes connected by inclusions, called a \emph{filtration of $K$}, denoted by 
$\Fcal:~K_1 \subseteq K_2 \subseteq \cdots K_m = K.$
Applying the homology functor to this sequence (with $\ZZ_2$ coefficients), we obtain a sequence of vector spaces (over field $\ZZ_2$) connected by linear maps induced from inclusions, which is called a persistence module; in particular, for any dimension $q\ge 0$, we have: 

\[\mathbb{P}\Fcal:~~~~  \homo_q(K_1) \to \homo_q(K_2) \to \cdots \to \homo_q(K_m).\]
%where 
%$\homo_q(K')$ stands for the $p$-th simplicial homology group of $K'$, and 
where maps are induced by inclusions. 
In our paper, we assume that the persistence module is indexed by a finite set $[1, m]$ instead of $\ZZ$. 

A special class of persistence modules is the so-called \emph{interval modules}. 
(i)  $I_i = \ZZ_2$ for any $\ell \in [s,t]$ and $I_i = 0$ otherwise; and (ii) $\nu^{i,j}$ is identity map for $s\le i \le j \le t$ and 0 map otherwise. We abuse the notation slightly and allow $t =\infty$, in which case the interval is really $[s, \infty)$. 
A pictorial version of an interval module is as follows: 
\[\cdots \to 0 \to \ZZ_2 \to \ZZ_2 \to \cdots \ZZ_2 \to 0 \to \cdots. \]

\paragraph{Persistence diagram.} It turns out that a given persistence module $\mathbb{V}$ can be uniquely decomposed into direct sums of interval modules (up to isomorphisms) $\mathbb{V} = \bigoplus_{[b,d]\in J} \mathbb{I}^{[b,d]}$, where $J$ is a multiset of intervals $J = \{ [b, d]\}$. 
We call $\bigoplus_{[b,d]\in J} \mathbb{I}^{[b,d]}$ the interval decomposition of $\mathbb{V}$. 
Again, note that the intervals in $J$ could be of two forms: $[b, d]$ for finite $b, d\in \ZZ$, and $[b, \infty)$; the former is called a \emph{finite interval}. Note that each interval $[b,d]$ can also be viewed as a point in $\reals^2$. 
Given a filtration $\Fcal$, its \emph{persistence diagram} $\dgm \Fcal$ is the multiset of points in $J$ where $\mathbb{P}\Fcal = \bigoplus_{[b,d]\in J} \mathbb{I}^{[b,d]}$ is the interval decomposition of $\mathbb{P}\Fcal$. Each point in $J$ is called a \emph{persistence point}. 
Assuming that we are given a monotone function $f: \ZZ \to \reals$, then the persistence of  $\pp = [b, d] \in \dgm \Fcal$ w.r.t. $f$ is defined as $\pers(\pp) = f(d) - f(b)$ \footnote{We note that in the literature, the persistence of a pair is often defined using some indices ($\mathbb{Z}$ or $\reals$) of the filtration. Here we decouple the two to make the presentation cleaner.}. To make the dependency on the function $f$ explicit, we now write the filtration together with this function as $\Fcal_f$, and the persistence diagram is denoted by $\dgm \Fcal_f$. 
For example, a common choice of $f$ in the literature is simply $f(i) = i$. 

\paragraph{Simplex-wise setting.} 
In the remainder of this paper, we assume that we are given a \emph{simplex-wise filtration $\Fcal$} of $K$, such that there is an ordering of all simplices in $K$, $\sigma_1, \ldots, \sigma_N$, and the filtration is given by: 
\begin{align}\label{eqn:simplexwiseF}
 \Fcal:  \emptyset = K_0 \subset K_1 \subset \cdots \subset K_N = K, ~~\text{where}~K_i:= \{\sigma_1, \ldots, \sigma_i \}. 
\end{align}
Suppose we are also given a monotone function $\rho: [1, N] \to \reals$ (i.e, $\rho(j) \ge \rho(i)$ for $j > i$), which we use to define the persistence of points in the persistence diagram $\dgm \Fcal_\rho$. (If no function $\rho$ is explicitly given, we take $\rho$ to be $\rho(i) = i$.)

Furthermore, note that for any $i$, $K_i$ is obtained by adding $\sigma_i$ to $K_{i-1}$. 
Let $\myo: K \to [1, N]$ be this bijection, where we set $\myo(\sigma_i) = i$. That is, $\myo(\sigma)$ in general is the index of simplex $\sigma$ in the ordered sequence of simplices that induce simplex-wise filtration $\Fcal$; or, the time it will be inserted into a complex (i.e $K_{\myo(\sigma)}$) in the filtration. 
Given this bijection, a function on the simplices in $K$ also gives rise to a function on $[1,N]$. 
In what follows, for convenience, we do not differentiate a function on simplices in $K$ and a function on $[1, N]$; that is, $\rho(\sigma) = \rho (\myo(\sigma))$, and if simplices are ordered as in Eqn (\ref{eqn:simplexwiseF}), then $\rho(\sigma_i) = \rho(i)$. If $\rho$ is defined on simplices in $K$, we also call it a \emph{simplex-wise function $\rho: K \to \reals$}. 

Given any persistence point $[b, d]\in \dgm \Fcal_\rho$ with $b, d\in [1, N]$, we say its corresponding \emph{persistence pair} is $(\sigma_b, \sigma_d)$ and it is necessary that $dim(\sigma_d) = dim(\sigma_b) + 1$. We set $\pers(\sigma_b) = \pers(\sigma_d) = \pers([b,d]) = \rho(d) - \rho(b)$. If $d = \infty$, then we say $\sigma_b$ is unpaired, and $\pers(\sigma_b) = \rho(\infty) := \infty$. 
%We will also write $\rho(\sigma_i) = \rho(i)$. 
Finally, consider each persistence pair $(\sigma, \tau)$, we say that $\sigma$ is \emph{positive} and $\tau$ is \emph{negative}, as the $q$-simplex $\sigma$ will create a new homology class that will become trivial (be killed) when the ($q$+$1$)-simplex $\tau$ is added to the filtration. 

A common way to induce a filtration is via a descriptor function $\rho: V(K) \to \reals$ given at vertices $V(K)$ of $K$. For simplicity of presentation, assume that $\rho$ is \emph{injective}. We can extend $\rho$ to a simplex-wise function $\rho: K \to \reals$ by setting $\rho(\sigma) = max_{v\in \sigma} \rho(v)$. 
Consider an ordering of simplices $\mathcal{S}_\rho: \sigma_1, \ldots, \sigma_m$ that is \emph{consistent with} $\rho$; i.e, (i) $\rho(\sigma_i) \le \rho(\sigma_j)$ for any $i\le j$ and (ii) for any simplex $\sigma_i$, its faces appear before it in the ordering. This order induces the so-called \emph{lower-star filtration} $\Fcal_\rho$ w.r.t. $\rho$. 
That is, assume $\rho(v_1) < \ldots < \rho(v_n)$.  Intuitively, we inspect the domain in increasing values of $\rho$ and the lower-star filtration is obtained by adding each vertex $v_i$ and its lower-star (simplices incident on $v_i$ with function value at most $\rho(v_i)$) in ascending order of $i$. 
The persistence diagram $\dgm_p \Fcal_\rho$ encodes birth and death of features during this course. In this case, we modify the persistence to reflect function values: For a persistence point $(b, d)\in \dgm_p \Fcal_\rho$, we set $\pers((b,d)) = \pers((\sigma_b, \sigma_d)) := |\rho(\sigma_d) - \rho(\sigma_b)|$. Features with large persistence survive for a long range of function values and are considered as more important w.r.t. $\rho$. 

\subsection{Discrete Morse Theory}
Below we very briefly introduce some concepts from discrete Morse theory, so that we can introduce both the original algorithm of \cite{WWL15} (to provide intuition) and the simplified algorithm of \cite{DWW18}. See \cite{For98,For01} for more detailed exposition of discrete Morse theory. 

We again consider the simplicial complex setting. 
Given a simplicial complex $K$, a \emph{discrete gradient vector} is a combinatorial pair of simplices $(\sigma^q, \tau^{q+1})$ where $\sigma$ is a face of $\tau$ of co-dimension $1$ (i.e, $\sigma$ is a vertex of an edge $\tau$, or an edge of a triangle $\tau$), and we sometimes include the superscript to make its dimension explicit. 
Given a collection $M(K)$ of such discrete gradient vectors over $K$, a \emph{V-path} is a sequence of simplices of alternating dimensions: 
$\sigma^q_1, \tau^{q+1}_1, \ldots, \sigma^q_\ell, \tau^{q+1}_\ell, \sigma^q_{\ell+1}$ such that for each $i\in [1, \ell]$,  we have (1) $(\sigma^q_i, \tau^{q+1}_i) \in M(K)$ and (2) $\sigma^q_{i+1}$ is a face of $\tau^{q+1}_i$. 
We say that a V-path as above is a \emph{non-trivial closed V-path} (or \emph{cyclic}) if $\sigma_1 = \sigma_{\ell+1}$; otherwise, it is \emph{acyclic}. 

\begin{definition}[Discrete Morse gradient vector field]
A collection of discrete gradient vectors $M(K)$ of $K$ is a \emph{discrete Morse gradient vector field}, or \emph{DM-vector field} for short, if (i) any simplex in $K$ is in at most one vector in $M(K)$; and (ii) no V-path in $M(K)$ is cyclic. 

A simplex in $K$ is \emph{critical} w.r.t. a DM-vector field $M(K)$ if it does not appear in any gradient vector in $M(K)$. 
\end{definition}

Now suppose we are given a critical edge $e$ in $M(K)$. The \emph{stable 1-manifold} of $e$ is the union of vertex-edge V-paths 
$v_1, e_1, \ldots, v_\ell, e_\ell, v_{\ell+1}$ such that $v_1$ is an endpoint of $e$, while $v_{\ell+1}$ is a critical vertex. 
Such stable 1-manifolds correspond to the "valley ridges" in a continuous function $f: \reals^d \to \reals$ (the graph of which can be viewed as a terrain), connecting index-1 saddles with minima. They are the opposite of "mountain ridges" (unstable 1-manifolds), connecting saddles to maxima and separating different valleys. 

Finally, we note that there is a \emph{Morse cancellation} operation that allows one to cancel a pair of critical simplices, and thus reduce both the number of critical simplices as well as the complexity of (un)stable 1-manifolds. 
In particular, a pair of critical simplices $\langle \sigma^q, \tau^{q+1}\rangle$ is \emph{cancellable} if there is a unique V-path 
$\sigma_{1}, \tau_{1}... , \sigma_{\ell}, \tau_{\ell}, \sigma_{\ell + 1} = \sigma^q$ in $M(K)$ such that $\sigma_1$ is a face of $\tau^{q+1}$.
The Morse cancellation operation will essentially invert the gradient vectors along this V-path and render $\sigma^q$ and $\tau^{q+1}$ no longer critical afterwards.

\subsection{Graph Reconstruction Algorithm for Density Field Based on Morse Theory}
\label{subsec:oldDM}
Below we first introduce the intuition behind the original discrete Morse based graph reconstruction algorithm from density field by \cite{2011MNRAS,WWL15} in the smooth setting. We will then describe the discrete setting, and its simplification \oldDM{} by \cite{DWW18}. 
First, assume we are given a smooth function $\rho: \mydomain \to \reals$ on a hypercube $\mydomain$ in $\reals^d$. View $\rho$ as a density function which concentrates around a hidden geometric graph (e.g, Figure \ref{fig:morse_neuron} (A) where $\mydomain \subset \reals^2$). Consider the graph of this function $\{(x, \rho(x) ) \mid x\in \mydomain\}$, which is a terrain in $\reals^{d+1}$ and which we will refer to as the terrain of $\rho$; see Figure \ref{fig:morse_neuron} (B). Intuitively, the "mountain ridge" of this terrain identifies the hidden graphs, as locally on the hidden graph, the density should be higher than points off it.  To capture these mountain ridges, one can use the so-called \emph{unstable 1-manifolds} of the function $\rho$ as in \cite{2011MNRAS,WWL15}. 

\begin{figure*}
\begin{center}
\includegraphics[width = 0.8\linewidth]{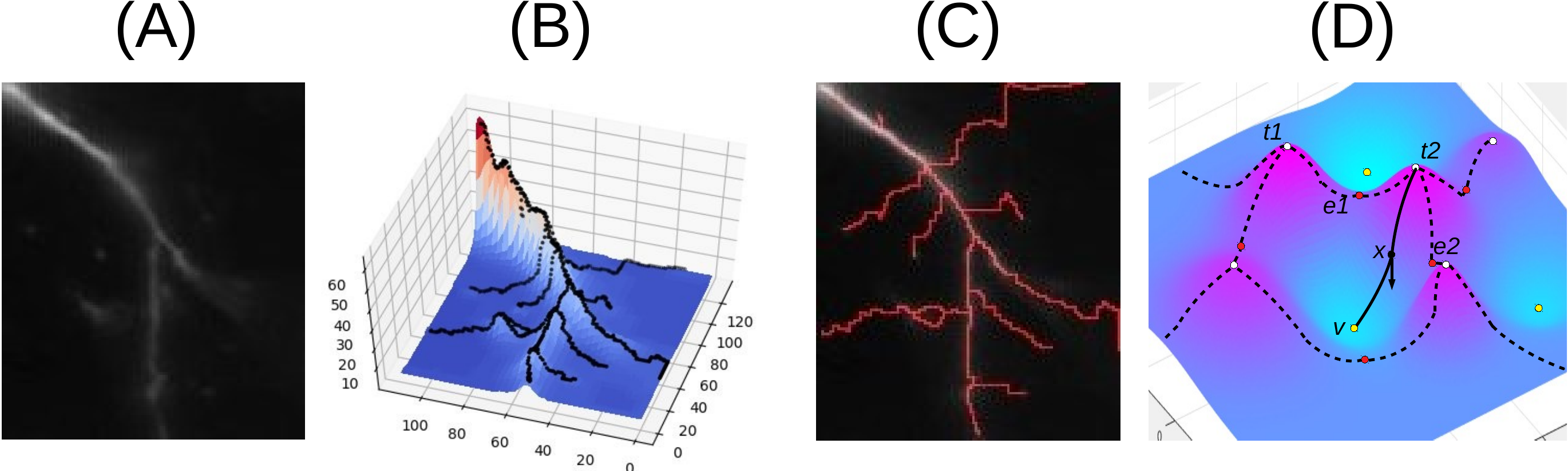}
\end{center}
\vspace*{-0.1in}
\caption{{\small A practical example of DM graph reconstruction. (A) The image (downloaded from \url{www.brainimagelibrary.org/})  contains neuronal branches that we aim to reconstruct. (B) View the image as a density function, we show the graph of this function, and mountain ridges of this terrain. 
(C) These ridges capture potential neuronal branches in the image in (A). (D) Gradient and the integral line passing $x$. Dashed curves are union of unstable 1-manifolds. 
}}
\label{fig:morse_neuron}
\end{figure*} 

Roughly speaking, given $\rho$, the gradient vector at $x\in \mydomain$, $\nabla \rho(x) = -[\frac{\partial \rho}{\partial \xx_1}(x), \ldots, \frac{\partial \rho}{\partial \xx_d}(x)]^T$, indicates the steepest descending direction of $\rho$ at $x$. See Figure \ref{fig:morse_neuron} (D). 
\emph{Critical points} of $\rho$ are points whose gradient vector vanishes. For a smooth function on $d$-D domain, non-degenerate critical points include minima, maxima, and $d$-$1$ types of saddle points. 
An integral line is intuitively the flow-line traced out by following the gradient direction at every point.
Flow-lines (integral lines) start and end (in the limit) at critical points. 
The \emph{unstable 1-manifold} of a saddle (of index $d$-$1$) is the union of flow-lines starting at some maximum and ending at this saddle. Intuitively, unstable 1-manifolds connect mountain peaks to saddles to peaks, separating different valleys (around minima), and thus can be used to capture mountain ridges. 

Hence one can compute the union of unstable 1-manifolds of $\rho$ as its graph skeleton. Furthermore, the density map $\rho$ may be noisy. To denoise the graph skeleton, previous approaches use persistent homology to keep only unstable 1-manifolds corresponding to "important" saddles. 

\paragraph{Algorithm in the discrete setting.} 
In the discrete setting imagine $K$ is the 2-skeleton of a domain $\Omega$ of interest, $\rho$ is a density function defined on $\Omega$ but is only accessible at vertices $V(K)$ of $K$, i.e., $\rho: V(K) \to \reals$. 
Algorithm \DMfirst($K, \rho: V(K) \to \reals, \delta$) will output a graph consisting of edges of $K$ capturing a graph skeleton of the density field $\rho$ by the following three steps:

\begin{itemize}
    \item (Step 1): Compute persistence pairing $\mathcal{P}$ induced by the  lower-star filtration w.r.t. -$\rho$.
    
    Specifically, we use $f = -\rho$ as it is easier to algorithmicly compute the discrete analog of "valley ridges" using discrete Morse theory than "mountain ridges" -- The valley ridges are the stable 1-manifolds (vertex-edge V-paths) for critical edges, and thus only 2-skeleton of input complex $K$ is needed. 
    To compute the importance of critical points in the simplicial setting when we are given $f: V(K) \to \reals$, we use the standard lower-star filtration to simulate the so-called sublevel-set filtration in the smooth case. In particular, given $f: V(K) \to \reals$, let $v_1 \ldots v_n$ be the set of vertices in $K$ sorted in non-decreasing order of $f$ values. Given any vertex $v_i \in V(K)$, its \emph{lower-star $\lst(v_i)$} consists of the set of simplices incident on $v_i$ spanned by only vertices from $V_i:= \{v_1, \ldots, v_i\}$. The lower-star filtration w.r.t. $f$ is the following: 
\begin{align}\label{eqn:lowstF}
\Khat_1 \subset \Khat_2 \subset \cdots \Khat_n  = K; ~~~\text{where}~ K_i = K_{i-1} \cup \lst(v_i). 
\end{align}
Equivalently, we can think that this filtration is induced by a simplex-wise function $\hat{f}: K \to \reals$ where $\hat{f}(\sigma) = \max_{\text{vertex}~ v~ \text{of}~ \sigma} f(v)$. 
\item (Step 2): Initialize vector field $M(K)$ to be the trivial one where all simplices are critical. Then in order of increasing persistence, for each pair $(\sigma, \tau) \in \mathcal{P}$ with $\pers(\sigma,\tau) \le \delta$, perform discrete Morse cancellation and update $M(K)$ if possible. 
Intuitively, this is to simplify and remove "not-important" critical points. 
    \item (Step 3): Output the graph $G_\delta = \bigcup_{e \in K, \pers(e) > \delta} \{$ stable 1-manifold of $e \}$. 
    
    In particular, we only consider critical edges
    that are "important" (i.e., $\pers > \delta$). Then we trace the valley ridges (stable 1-manifolds) connecting them to minima. These minima - which have persistence greater than $\delta$ - are the topographically prominent peaks of $\rho$.
\end{itemize}

\paragraph{Simplified algorithm.} It turns out that algorithm \DMfirst() can be significantly simplified \cite{DWW18}. In particular, one does not need to explicitly maintain any discrete Morse gradient vector field at all. See Algorithm \oldDM() below. 

\begin{algorithm2e}
\SetAlgoLined
\KwIn{ Triangulation $K$, density function $\rho: V(K) \to \reals$, persistence  
threshold $\delta$ } 
\KwOut{ a graph skeleton $G_\delta$}

    (Step 1) Compute persistence pairing $\mathcal{P}$ induced by the lower star filtration w.r.t. -$\rho$,
    
    (Step 2) Set $\Tcal_\delta:= \{ e \in E \mid e$ is negative and $\pers(e) \le \delta \}$
    
    \quad For each component (tree) $T$ in $\Tcal_\delta$, set its root to be $r(T):= \mathrm{argmin}_{v\in T}$ -$\rho(v)$.
    
    (Step 3) Let $\pi_T(x,y)$ be the tree path from $x$ to $y$ in a tree $T$. Output:  
    \begin{align}\label{eqn:outputG}
    G_\delta = \bigcup_{e =(u,v), \pers(e) > \delta} \{ e \cup \pi_{T_{i_1}}(u, r(T_{i_1})) \cup \pi_{T_{i_2}}(v, r(T_{i_2})) \mid u\in T_{i_1},  v\in T_{i_2}~\text{in}~\Tcal_\delta\}. \end{align} 

 \caption{\oldDM($K, \rho, \delta$)}
 \label{alg:oldDMalgo}
\end{algorithm2e}

In particular, in (Step 3) above, we only consider critical edges with $\pers > \delta$, and their stable 1-manifolds turn out to be the union of tree paths as specified in Eqn (\ref{eqn:outputG}). Note that (Step 2, 3) can be implemented in time linear to the number of vertices and edges in $K$. 

\section{Generalized Algorithm and Optimality}
\label{sec:generalalg} 

Now suppose instead of a triangulation of a $d$-D domain $\Omega$, we have an arbitrary simplicial complex $K$ -- our algorithm only needs its 2-skeleton $K=(V,E,T)$.
Suppose further that there is a simplex-wise function $\newrho: K \to \reals$.
Let $\myorder_\newrho:=\langle \sigma_1, \ldots, \sigma_N \rangle$ be an ordered sequence of simplicies of $K$ that is \emph{consistent with $\newrho$} (see the end of Section \ref{subsec:PHintro}),
and let $\Fcal_\newrho$ be the simplex-wise filtration of $K$ induced by this order $\myorder_\newrho$. 
(We will describe in Section \ref{sec:PCD} how to set up this filtration for graph skeletonization from PCDs.) 
We now generalize algorithm \oldDM() to the following \newDM(), where essentially, only (Step 1) differs by taking an arbitrary simplex-wise filtration $\Fcal_\newrho$, which we state in Algorithm \ref{alg:generalizedDMalgo} for clarity. 

\begin{algorithm2e}
\KwIn{Arbitrary simplex-wise filtration $\Fcal_\newrho$ of a simplicial complex $K=(V,E,T)$, threshold $\delta$ } 
\KwOut{ A reconstructed graph $G_\delta$ }

    (Step 1) Compute persistence pairings w.r.t. $\Fcal_\newrho$
    
    (Step 2) + (Step 3): same as in alg. \oldDM{}
    
 \caption{\newDM($K, \Fcal_\newrho, \delta$)}
 \label{alg:generalizedDMalgo}
\end{algorithm2e}

It is easy to verify that the original \oldDM($K, \rho, \delta)$ algorithm 
is a special case of the above algorithm, where we set $\Fcal_\newrho$ in \newDM($K,\Fcal_\newrho, \delta$) to be the lower-star filtration induced by the vertexwise function $\rho'= -\rho: V(K) \to \reals$; specifically, for any simplex $\sigma \in K$, set $\newrho(\sigma):= \max_{v\in \sigma} -\rho(v)$. 
The difference between our \newDM() algorithm and the original algorithm is rather minor. However, we will see that this change of perspective (from density-function based view to arbitrary filtration-based view) significantly broadens the applicability of this algorithm. In particular, in Section \ref{sec:PCD} we will show how this generalized algorithm can be combined with weighted Rips sparsification strategy to reconstruct a hidden graph skeleton of high-dimensional points data. 
But first, in what follows, we provide some characterization of the graph skeleton output by \newDM{}. Specifically, we show that the output of \newDM() contains the so-called \emph{lex-optimal} cycle basis of $K$ w.r.t. important 1D homological features  in $\dgm_1 \Fcal_\newrho$. 
To make this statement more precise, we first introduce some notations, following \cite{CLV20,DeyHM20,Wuetal17}. 
Intuitively, a $1$-cycle is a collection of edges forming one or multiple closed loops; and 
a $d$-cycle is a $d$-D analog of it. 
\begin{definition}[Persistent cycles \cite{DeyHM20}]
\label{def:percycle}
Let $\Fcal$ be a simplexwise filtration of $K$ induced by the ordered sequence of simplices $\sigma_1, \ldots, \sigma_N$, and $\dgm_q \Fcal$ its resulting $q$-th persistence diagram. 
Given a point $\pp= [b, d] \in \dgm_{q} \Fcal$, a $q$-cycle $\gamma$ is \emph{a persistent $q$-cycle w.r.t. $\pp$} if
(i) if $d \neq \infty$, $\gamma$ is a cycle in $K_b$ containing $\sigma_b$, and $\gamma$ is not a boundary in $K_{d-1}$ but becomes a boundary in $K_d$; and (ii) otherwise if $d=\infty$, then $\gamma$ is a cycle in $K_b$ containing $\sigma_b$. 

Given a subset $D = \{\pp_1, \ldots, \pp_r\} \subseteq \dgm_{q} \Fcal$ with $r=|D|$, we say that a set of cycles $\{\gamma_1, \ldots, \gamma_r \}$ form a \emph{persistent cycle-basis} for $D$ if $\gamma_i$ is a persistence cycle w.r.t. $\pp_i$ for all $i\in [1,r]$. 

\end{definition}

Roughly speaking, a persistent cycle $\gamma$ w.r.t. a persistence point $\pp = [b,d]$ is created at $b$ and killed at $d$, and can be thought of a representative of the homological feature captured by point $\pp \in \dgm \Fcal$. A persistent cycle basis w.r.t. $D \subset \dgm \Fcal$ corresponds to representative cycles captured by points in $D$.
More specifically, 
given a cycle $\gamma$, let $[\gamma]_{K_i}$ denote the homology class of $\gamma$ in complex $K_i$. 
The following result from \cite{DeyHM20} intuitively says that a persistence cycle-basis for $\dgm_q \Fcal$ essentially generates the interval decomposition of persistence module $\mathbb{P}\Fcal$. 
\begin{claim}[\cite{DeyHM20}]
Let $\{\gamma_1, \ldots, \gamma_g \}$ be a persistence cycle-basis for $\dgm_q \Fcal = \{ \pp_1, \ldots, \pp_g\}$ and $g = |\dgm_q \Fcal|$. 
Then $\mathbb{P}\Fcal = \bigoplus_{\pp_\ell \in \dgm_q \Fcal} \mathbb{I}^{\pp_\ell}$, where the interval module $\mathbb{I}^{\pp_\ell} = \{I_i \overset{\tiny{\nu_{i,j}}}{\longrightarrow} I_j \}_{i\le j}$ is generated by $\gamma_\ell$ in the sense that $I_i = [\gamma_\ell]_{K_i}$. 
\end{claim}

Lexicographic optimal cycles are introduced in \cite{CLV19,CLV20}. We will extend them to the persistence version. 
Given a simplex-wise filtration $\Fcal$ of $K$ induced by an ordering of simplices $\sigma_1, \ldots, \sigma_N$, we set $\myo(\sigma)$ as the order it appears in $\Fcal$; i.e, $\myo(\sigma_i) = i$.  
%, recall we order all simplices by the order they appear in $\Fcal$: $\sigma_1, \ldots, \sigma_N$. 
%We set $\myo(\sigma)$ as its order in $\Fcal$; i.e, $\myo(\sigma_i) = i$.   

\begin{definition}[Lexicographic order \cite{CLV20}]
Given two $q$-cycles $C_1, C_2 \in \myC_q(K)$, 
we say that $C_1 \mypreceq C_2$ if either (i) $C_1 + C_2 = 0$ or (ii) otherwise, the simplex $\sigma_{max}:= \argmax_{\sigma\in C_1+C_2} \myo(\sigma)$ is from $C_2$. 
If (ii) holds, we say that $C_1 \myprec C_2$, i.e., $C_1$ is \emph{smaller than $C_2$ in lexicographic order}. 
Intuitively, $C_1 \mypreceq C_2$ if simplices in $C_1$ comes "earlier" than $C_2$ in the filtration order. 
\end{definition}

\begin{definition}[Lex-optimal persistent cycles]% and cycle-basis] \label{def:lexoptpercycles} 
Given a persistence point $\pp=[b, d]\in \dgm_{q} \Fcal$, a $q$-cycle $\gamma$ is a \emph{lexicographic-optimal (lex-opt for short) persistent cycle w.r.t. $\pp$} %(or \emph{\lexoptper{} cycle} in short), 
if (i) $\gamma$ is a persistent cycle w.r.t. $\pp$; and 
(ii) among all persistence cycles w.r.t. $\pp$, $\gamma$ has the smallest lexicographic order. 
We say that $\Gamma = \{\gamma_1, \ldots, \gamma_r\}$ forms a \emph{lex-optimal persistent cycle basis} for a multiset $D = \{\pp_1, \ldots, \pp_r\} \subseteq \dgm_{q} \Fcal$ if $\gamma_i$ is a lex-optimal persistence cycle w.r.t $\pp_i$ for all $i\in [1, r]$. 
\end{definition} 

Given a $q$-th persistence diagram $\dgm_{q} \Fcal$ and a threshold $\delta$, let $\dgm_{q} \Fcal(\delta) \subseteq \dgm_q\Fcal$ denote the subset of points in $\dgm_{q} \Fcal$ whose persistence is larger than $\delta$ (intuitively, these corresond to important features). Our first main result is the following theorem. 

\begin{theorem}\label{thm:lexOPHC}
(i) $G_\delta$ as constructed w.r.t. a simplex-wise filtration $\Fcal_\rho$ contains a lex-optimal persistence cycle basis for $\dgm_{1} \Fcal_\rho (\delta)$, and (ii) the first Betti number of $G_\delta$ equals $|\dgm_1 \Fcal_\rho(\delta)|$. 
\end{theorem}

The above theorem suggests that the output graph $G_\delta$ by our algorithm \newDM() contains the "best" loops whose homology classes have large persistence and whose edges come \emph{as early as possible in the filtration}. In particular, imagine that important edges or more faithful edges come early in the filtration, then the output graph contains those loops with large persistence ($>\delta$) and formed by more faithful edges whenever possible. 
In the graph reconstruction from PCDs application in the next section, intuitively, if edges from high-density region come into the filtration first, then the resulting output graph will use such edges whenever possible. See Figure \ref{fig:gaussian_circle_results} (A) to (D).  

\section{Proof of Theorem 3.5}
\label{sec:lexopt-proof}

We assume that $K$ is connected. If it is not, then we will perform the following arguments to each connected component of $K$.  
Now recall that $\Tcal_\delta := \{e\in E \mid e ~\text{is negative and}~\pers(e) \le \delta \}$ consists of all negative edges with persistence at most $\delta$ (from (Step 2) of algorithm \newDM{} in the main paper). 
It is shown in \cite{DWW18} that $\Tcal_\delta$ consists of a set of trees. 
Set

\begin{align*} 
E^+_\delta &:= \{ e\in E \mid e ~\text{is positive and}~ \pers(e) > \delta \},~\text{and}~ \\
E^-_\delta &:= \{ e\in E \mid e ~\text{is negative and}~ \pers(e) > \delta \}.
\end{align*}
Set $\Ghat_\delta = \Tcal_\delta \bigcup G_\delta$, where $G_\delta$ is the output of algorithm \newDM{}. Furthermore, by construction, $G_\delta$ consists of edges in $E^-_\delta \cup E^+_\delta$ together with a set of tree paths in $\Tcal$ (recall Eqn (\ref{eqn:outputG}) in Algorithm \ref{alg:oldDMalgo}, which is the same as the construction for algorithm \newDM{}). It follows that 

\begin{align}\label{eqn:Ghat} 
\Ghat_\delta &= \Tcal_\delta \bigcup G_\delta= \Tcal_\delta \bigcup E^-_\delta \bigcup E^+_\delta. 
\end{align} 
We prove Theorem 3.5 in two steps, laid out in the following two lemmas.

\begin{figure*}
\begin{center}
    \begin{tabular}{ccc}
    \fbox{\includegraphics[height=3cm]{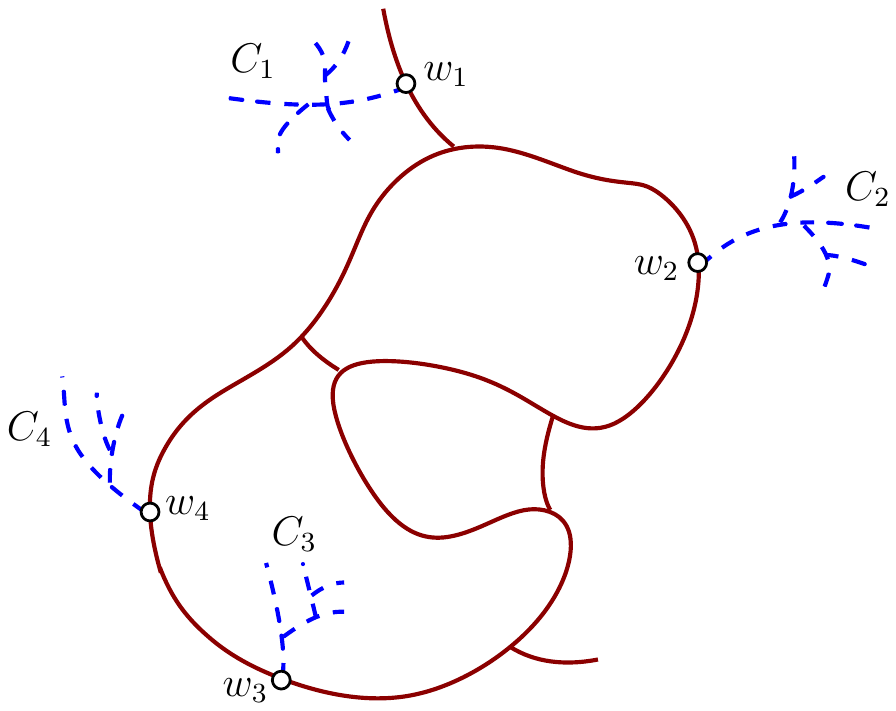}}  &
    \fbox{\includegraphics[height=2.5cm]{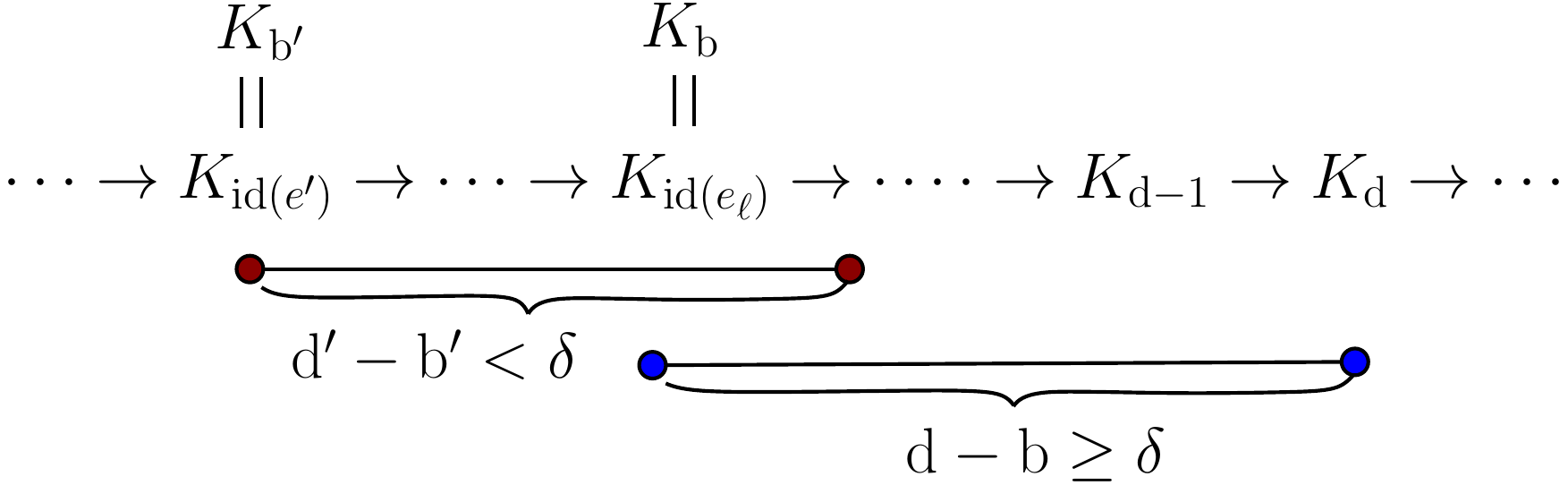}}\\
    (a) & (b)
    \end{tabular}
\end{center}
\vspace*{-0.1in}
\caption{(a) The solid red curve is $G_\delta$, while dashed trees are components in $\Ghat_\delta \setminus G_\delta$. The closure of each component $C_i$ connects to $G_\delta$ at one point $w_i$, and thus its closure can deformation retracts to $w_i\in G_\delta$.
(b) As $\pers(e') = \dd' - \bb' \le \delta$, and $\pers(\mye_\ell) = \dd - \bb > \delta$, and $\bb' (= \myo(e')) \le \bb$, it then follows that the persistent cycle $C' = \pi(u,v) + e'$ must become a boundary in the simplicial complex $K_{\dd - 1}$, which in turn leads to that $[\gamma'] = [\gamma^*]$ in $K_{\dd-1}$.
\label{fig:sequencecycle}}
\end{figure*}

\begin{lemma}\label{lem:Ghatgood} 
Statements (i) and (ii) in Theorem 3.5 holds for $\Ghat_\delta$. That is: (i') $\Ghat_\delta$ constructed w.r.t. a simplex-wise filtration $\Fcal_\rho$ contains a lex-optimal persistence cycle basis for $\dgm_{1} \Fcal_\rho (\delta)$, and (ii') the first Betti number of $\Ghat_\delta$ equals $|\dgm_1 \Fcal_\rho(\delta)|$. 
\end{lemma}

\begin{lemma}\label{lem:contractible}
$\Ghat_\delta$ deformation contracts to $G_\delta$. 
\end{lemma} 

Our theorem then follows from these two lemmas. Specifically, we will use the graph $\Ghat_\delta$ as a proxy: Lemma \ref{lem:Ghatgood} states that the desired results hold for $\Ghat_\delta$. 
Lemma \ref{lem:contractible} then relates $\Ghat_\delta$ to $G_\delta$. In particular, as both $\Ghat_\delta$ and $G_\delta$ are graphs, this lemma implies that any simple cycle in $\Ghat_\delta$ must be present in $G_\delta$ as well. 
Theorem 3.5 then follows. 
What remains is to prove these two lemmas, which we present in the two subsections that follow.  

\subsection{Proof of Lemma \ref{lem:Ghatgood}}

Let $\dgm_1 \Fcal_\rho (\delta) = \{ \pp_1, \ldots, \pp_g \}$. 
By the definition of positive and negative edges, we know: 
\begin{itemize}
\item (C1). $\Tcalhat = \Tcal_\delta \cup E^-$ is a spanning tree of $K$. 
\item (C2). By the definition of positive edges, $E^+_\delta$ contains exactly those edges whose addition create the persistence points in $\dgm_1 \Fcal_\rho(\delta)$. In other words, $g = |E^+_\delta|$ and we can order edges in $E^+_\delta = \{ \mye_1, \ldots, \mye_g\}$ so that for any $\ell\in [1, g]$,  $\pp_\ell = [\myo(\mye_\ell), d_\ell]$: i.e, the birth-time of $\pp_\ell$ corresponds to the insertion of edge $\mye_\ell$ in the simplicial complex $K_{\myo(\mye_\ell)}$. 
\end{itemize} 
Furthermore, the addition of each positive edge $\mye_\ell \in E^+_\delta$ creates a cycle in the spanning tree $\Tcalhat$ (as $\mye_\ell$ is not a tree edge), 
As $\Ghat_\delta = \Tcalhat \cup E^+$, we thus have $\beta_1 (\Ghat_\delta) := rank(\homo_1(\Ghat_\delta))$ is the same as $g = |\dgm_1 \Fcal_\rho(\delta)|$. This proves part (ii') in Lemma \ref{lem:Ghatgood} for the graph $\Ghat_\delta$. 

We now prove part (i') of Lemma \ref{lem:Ghatgood}. 
Consider any $\mye_\ell \in E^+_\delta$, and let $\gamma^*$ denote a \lexoptper{} cycle of the corresponding persistence point $\pp_\ell = [\bb, \dd]$ (where $\bb = \myo(\mye_\ell)$). By Definitions 3.1 and 3.3 in the main paper,  
$\gamma^*$ necessarily contains $\mye_\ell$, and all other edges in $\gamma^*$ have an index smaller than $\myo(\mye_\ell)$. 
We will next prove that $\gamma^*$ is in $\Ghat_\delta$, that is, treating a cycle (under $\mathbb{Z}_2$ coefficients) as a set, $\gamma^* \subseteq \Ghat_\delta$. 

In particular, take any edge $e' \in \gamma^*$ with $e' \neq \mye_\ell$, we will show that $e'\in \Ghat_\delta$. 
\begin{itemize}
\item If $e'$ is negative, then this is trivially true as $e' \in \Tcalhat \subseteq \Ghat_\delta$. 
\item If $e'$ is positive but with persistence $\pers(e') > \delta$, then it is also true as $e'\in E^+_\delta \subseteq \Ghat_\delta$. 
\item So what remains is the case when $e' = (u, v)$ is positive but with $\pers(e') \le \delta$. However, we will show that this case cannot happen, which implies that $e' \in \Ghat_\delta$.

Assume this case happens for edge $e'$. Then let $C_{e'} (= \pi(u,v) + e') \subset K_{\myo(e')}$ be a persistent cycle w.r.t. the persistence point $[\bb' = \myo(e'), \dd']$ generated by $e'$.  
First, as the path (1-chain) $\pi(u,v)$ is contained in $K_{\myo(e')}$, all edges in $\pi(u,v)$ have an index less than that of $e'$. 
This means that the cycle $\gamma' = \gamma^* - e' + \pi(u,v)$ is necessarily smaller than $\gamma^*$ in lexicographic order. 
We now claim that $\gamma'$ is also a persistent cycle w.r.t. the persistence point $\pp_\ell =[\bb, \dd]$ corresponding to the positive edge $\mye_\ell$. 

Indeed, as $\gamma^*$ is a persistent cycle w.r.t. $\pp_\ell$, we know that $\myo(e') < \myo(\mye_\ell) = \bb$. Recall that the persistence point corresponds to the positive edge $e'$ is $[\bb' = \myo(e'), \dd']$.  
As $\pers(\mye_\ell) = \rho(\dd) - \rho(\bb) > \delta$ while $\pers(e') = \rho(\dd') - \rho(\bb') \le \delta$, it then follows that $\dd' < \dd$. (See Figure \ref{fig:sequencecycle} (b) for illustrations of these notations.) 
Hence we know that it is necessary that the cycle $\pi(u,v) + e'$ becomes boundary in $K_{\dd-1}$. In other words, in $K_{\dd-1}$, the two cycles $\gamma^*$ and $\gamma'$ are homologous. It is then easy to verify that $\gamma'$ must be a persistent cycle for $\pp_\ell$ as well. 

Since $\gamma'$ is also a persistent cycle for $\pp_\ell$ and is lexicographically smaller than $\gamma^*$, this contradicts our assumption that $\gamma^*$ is a \lexoptper{} cycle for $\pp_\ell$. Hence no positive edge $e' \in \gamma^*$ with $\pers(e') < \delta$ can be in $\gamma^*$. 

\end{itemize}

By the above case analysis, any edge $e'\in \gamma^*$ must be in $\Ghat_\delta$. 
It then follows that $\gamma^* \subseteq \Ghat_\delta$. 
As this argument holds for any edge in $E^+_\delta$, we thus have proven (i').
This finishes the proof of Lemma \ref{lem:Ghatgood}. 

\subsection{Proof of Lemma \ref{lem:contractible}}
\label{appendix:lem:contractible}

First, by construction of $\Ghat_\delta$ (Eqn (\ref{eqn:Ghat})), we have that $G_\delta \subseteq \Ghat_\delta$, and all edges in $\Ghat_\delta \setminus G_\delta$ must come from $\Tcal_\delta$. 
Now recall $\Tcalhat = \Tcal_\delta \bigcup E^-$, which is a spanning tree of $K$. 
Given an arbitrary tree $T$ and two nodes $u,v\in T$, let $\pi_T(u,v)$ denote the unique tree path from $u$ to $v$ in $T$. We have the following simple claim. 
\begin{claim}\label{claim:onetree}
Given any rooted tree $T$ with root $r(T)$ and two nodes $u, v \in T$, we have that $\pi_T(u,v) \subseteq \pi_T(u, r(T)) \cup \pi_T(v, r(T))$. 
\end{claim}
\begin{proof}
If $u$ and $v$ have ancestor / descendent relation, say $u$ is ancestor of $v$, then it is clear that $\pi_T(u,v) \subseteq \pi_T(v, r(T))$, and the claim then follows. 
Otherwise, let $w$ be the common ancestor of $u$ and $v$. It can again be verified that in this case, $\pi_T(u,w) \subseteq \pi_T(u, r(T))$, $\pi_T(v, w) \subseteq \pi_T(v, r(T))$, while $\pi_T(u,v) = \pi_T(u,w) \circ \pi_T(w, v)$. The claim thus follows. 
\end{proof}

$\Tcal_\delta$ is a spanning forest of vertex set $V$. Given any vertex $v\in V$, suppose it is in the tree $T \in \Tcal_\delta$. We denote $\mypath(v):= \pi_T(v, r(T))$ to be the path from $v$ to the root $r(T)$ of $T$. 
Recall that $G_\delta$ is constructed by, for any edge $e = (u,v) \in E_\delta^- \cup E_\delta^+$, adding $e \cup \mypath(u) \cup \mypath(v)$ into $G_\delta$. 

\begin{lemma}\label{claim:onecycle} 
For each edge $\mye = (u,v) \in E^+_\delta$, set $\gamma = e \cup \pi_{\Tcalhat}(u,v)$. Then the cycle $\gamma$ must be contained in  $G_\delta$.
\end{lemma}
\begin{proof}
Consider the path $\pi = \pi_\Tcalhat(u,v)$: it will be broken into $k\ge 0$ maximally connected pieces from $\Tcal_\delta$, connected by edges in $E^- \cup E^+$. 
If $k = 0$, we are done, because this means that $u, v$ are contained in the same tree $T$ in $\Tcal$, and it then follows from Claim \ref{claim:onetree} that
\[\pi = \pi_T (u, v) \subseteq \mypath(u) \cup \mypath(v) \subseteq G_\delta . \] 

\begin{figure}[htbp]
\begin{center}
    \begin{tabular}{ccc}
    \includegraphics[height=3cm]{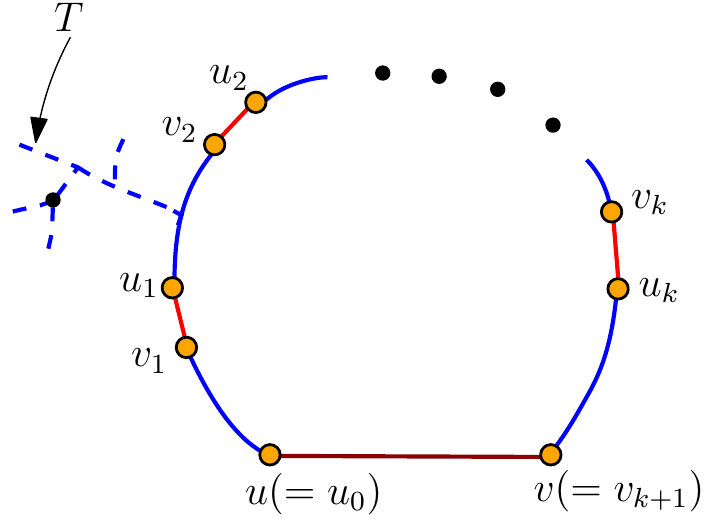} & 
        \includegraphics[height=2.8cm]{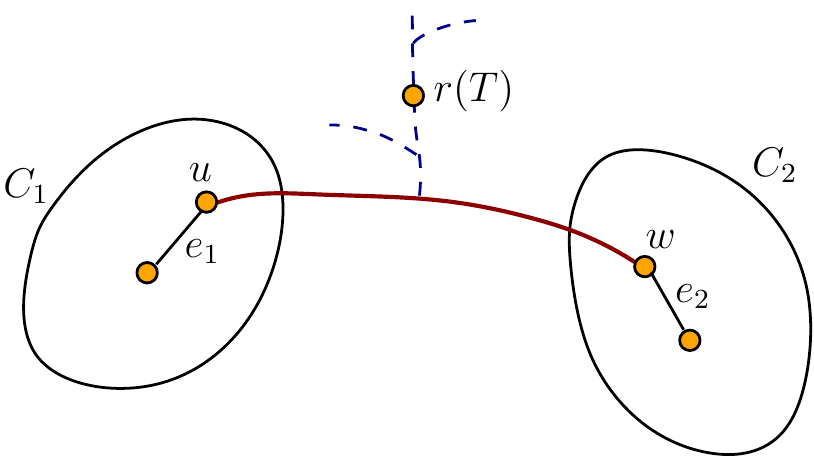} &
    \includegraphics[height=3cm]{figures/brokencycle2.pdf} \\
    (a) & (b) & (c)
    \end{tabular}
\end{center}
\vspace*{-0.08in}
\caption{(a) The path $\pi = \pi(u,v)$ is broken into $k+1$ pieces, each of which (blue subcurves) is a maximal connected component in $\pi \cap \Tcal_\delta$, while the connecting edges (red edges $(v_i, u_i)$'s) must come from $E_\delta^- \cup E_\delta^+$. (b) The solid red path is the tree path $\pi = \pi_T(u, w)$ connecting $u$ and $w$ in the tree $T$ from the spanning forest $\Tcal_\delta$. The root of $T$ is $r(T)$. (c) The solid red curve is $G_\delta$, while dashed trees are components in $\Ghat_\delta \setminus G_\delta$. The closure of each $C_i$ connects to $G_\delta$ at one point $w_i$, and thus its closure can deformation retract to $w_i$. }
\label{fig:brokencurve} 
\end{figure}
So assume that $k> 0$, and the edges connecting these pieces are $e_1 = (v_1, u_1), \ldots, e_k = (v_k, u_k)$ from $u$ to $v$ along $\pi$; see Figure \ref{fig:brokencurve} (a). Obviously, for each $i\in [1, k]$, $e_i \notin \Tcal_\delta$ and $e_i \in E_\delta^- \cup E_\delta^+$. 
Set $u_0 = u$ and $v_{k+1} = v$. It then follows that for any $i\in [0, k]$, $u_i$ is connected to $v_{i+1}$ within some tree, say $T \in \Tcal_\delta$. 
By Claim \ref{claim:onetree} that the portion of $\pi$ from $u_i$ to $v_{i+1}$ must be contained in $path(u_i) \cup path(v_{i+1})$. 
Applying this for all $i \in [0,k]$, it follows that 
\begin{align*}
    \pi \subseteq \big(\bigcup_{i\in [0,k]} (path(u_i) \cup path(v_{i+1}) \big)~ \bigcup~ \big( e_1 \cup e_2 \cdots \cup e_k\big) 
    \\
    = \big( path(u) \cup path(v) \big) ~\bigcup~ \big( e_1 \cup path(v_1) \cup path(u_1) \big) & \bigcup \cdots \bigcup ~\big( e_k \cup path(v_k) \cup path(u_k) \big). 
\end{align*}
As all edges $e_1, \ldots e_k$ and $e$ are all in $E^-\cup E^+$, it then follows that $\pi \subseteq G_\delta$ and thus $\gamma = e \cup \pi \subseteq G_\delta$. 
\end{proof}

\begin{lemma}\label{claim:bettinumbers} 
$\beta_0(G_\delta)  = \beta_0(\Ghat_\delta)$, and $\beta_1(G_\delta) = \beta_1(\Ghat_\delta)$.
\end{lemma}
\begin{proof}
That $\beta_1(G_\delta) = \beta_1(\Ghat_\delta)$ follows immediately from Lemma \ref{claim:onecycle}. We now prove that $G_\delta$ and $\Ghat_\delta$ also has the same number of connected components. 
Note that we have already assumed that $K$ is connected, and thus $\Ghat_\delta$ is connected as it contains a spanning tree $\Tcalhat$ of $K$. 
So what remains is to show that $G_\delta$ is connected. 

Assume $G_\delta$ is not connected, and let $C_1, C_2$ be two components of $G_\delta$. Recall that $G_\delta$ is constructed by the union of paths $\bigcup_{e=(u,v)\in E_\delta^- \cup E_\delta^+} \big(e \cup \mypath(u) \cup \mypath(v) \big)$. 
Let $e_1 \in C_1$ be an arbitrary edge from $C_1 \cap (E_\delta^- \cup E_\delta^+)$: Note that such an edge must exist, as otherwise $C_1$ will not be in $G_\delta$. 
Similarly, let $e_2 \in C_2 \cap (E_\delta^- \cup E_\delta^+)$. 
Let $u$ be an endpoint of $e_1$ while $w$ be an edge point of $e_2$. 
We know that $u$ and $w$ are connected in $\Tcal_\delta$ by path $\pi = \pi_{\Tcal_\delta}(u, w)$.  
We now claim that this path must be in $G_\delta$; which contradicts with our assumption that $C_1$ and $C_2$ are two connected components of $G_\delta$. Hence our assumption is wrong, and $G_\delta$ must be connected as well, which finishes the proof of the claim. 

What remains is to show that the path $\pi = \pi_{\Tcal_\delta}(u, w)$ as described above must be in $G_\delta$. 
Let $T\in \Tcal_\delta$ be the tree in $\Tcal_\delta$ containing path $\pi$, and let $r_T$ be its root. We now perform a case analysis based on the location of $r_T$ w.r.t. $u$ and $w$. (Case 1): $u$ is an ancestor of $w$ in $T$; (case 2): $w$ is an ancestor of $u$ in $T$; and (case 3): otherwise. See Figure \ref{fig:brokencurve} (b) for an illustration of (case 3). We first prove that $\pi \subseteq G_\delta$ for (case 3).
In this case, we have that $\mypath(u) \cup \mypath (w)$ is a superset of $\pi$, that is, $\pi\subseteq \mypath(u) \cup \mypath(w)$. Furthermore, since both $e_1, e_2 \in E^-_\delta \cup E^+_\delta$, by construction of $G_\delta$, $\mypath(u) \subseteq G_\delta$ and $\mypath(w) \subseteq G_\delta$. It then follows that $\pi \subseteq G_\delta$. 
Using a similar argument, one can show that $\pi \subseteq G_\delta$ for (case 1) and (case 2) as well. 

Putting everything together, we have that $G_\delta$ is connected and thus $\beta_0(G) = \beta_0(\Ghat_\delta)$. This finishes the proof of the lemma. 
\end{proof}

Now let $C_1, \ldots, C_s$ be the components of $\Ghat_\delta \setminus G_\delta$, and for each $i\in [1, s]$, let $\overline{C}_i$ be the closure of $C_i$. 
We claim that $\overline{C}_i \setminus C_i$ can contain only one vertex, say $w_i$. See Figure \ref{fig:brokencurve} (c). 
Indeed, as $\Ghat_\delta \setminus G_\delta \subseteq \Tcal_\delta$, each $\overline{C}_i$ is simply connected (i.e, it is a subtree of some tree in $\Tcal_\delta$). 
Suppose $\overline{C}_i \setminus C_i$ contains at least two vertices, say $w$ and $w'$. As $\overline{C}_i$ is connected, there is a path $\pi_{\overline{C}_i}(w, w')$ connecting $w$ to $w'$ in $\overline{C}_i$. 
On the other hand, as $G_\delta$ is connected (Lemma \ref{claim:bettinumbers}), there is another path $\pi_{G_\delta}(w, w')$ connecting $w$ and $w'$. 
This gives rise to a cycle $\gamma = \pi_{\overline{C}_i}(w,w') \cup \pi_{G_\delta}(w,w')$ in $\Ghat_\delta$, and this cycle is not in $G_\delta$. This however contradicts to what we just proved that $\beta_1(G_\delta) = \beta_1(\Ghat_\delta)$. 
Hence this cannot happen. 

Hence $\overline{C}_i$ can only connect to $G_\delta$ via one point $w_i$ as illustrated in Figure \ref{fig:brokencurve} (c). 
It then follows that $\Ghat_\delta$ deformation retracts to $G_\delta$ by contracting each subtree $\overline{C}_i$ to the point $w_i$. This finishes the proof of Lemma \ref{lem:contractible}. 

\section{PCD Algorithm via Sparse Weighted-Rips}
\label{sec:PCD}

Given a PCD $P \subset \reals^d$, we now wish to compute a graph skeleton of $P$. Our algorithm can be easily extended to the case where these points $P$ are not embedded but with only pairwise distances (or similarity) given. 

\paragraph{A baseline approach.}
A natural approach is to (i) build a simplicial complex $K$ from $P$ to "approximate" the space behind $P$, (ii) estimate a density function $\rho$ at $P= V(K)$, and (iii) then perform algorithm \oldDM.
A reasonable choice for $K$ is the so-called Rips complex $\rips^r(P) := \{ (p_{i_0}, \ldots, p_{i_k}) \mid \|p_{i_j} - p_{i_{j'}}\|\le r \}$: Intuitively, an edge $(p, q)\in \rips^r(P)$ if the distance between points $p,q\in P$ is at most $r$. A triangle is in $\rips^r(P)$ if all three edges are in, and similarly for higher-dimensional simplices. 
However, we only need 2-skeleton of $\rips^r(P)$, which we still denote by $\rips^r(P)$. The estimated density of a point is determined by summing the distances under a Gaussian kernel to each of its KNN for some k.
We refer to this algorithm as \mybaseline{} where we use the $\rips^r(P)$ as choice of complex $K$, that is, we perform \oldDM($\rips^r(P), \rho, \delta$). 

\paragraph{Challenges with \mybaseline.}
This \mybaseline{} approach faces several challenges. (C-1) It is usually hard to choose the right radius $r$ and the topology of $\rips^r(P)$ crucially decides the final output graph: see Figure \ref{fig:gaussian_circle_results}, where if $r$ is too small, the shape is not yet captured by $\rips^r(P)$; for larger $r$, there can be  spurious topological features (extra loops) in $K$ which cannot be simplified by persistence (as these loops are generated by edges with infinity persistence). 
There is also the issue that even if one has found a radius $r$ value such that $\rips^r(P)$ can provide the correct topology, the geometry of the graph skeleton computed by this baseline algorithm may lose resolution (e.g, Figure \ref{fig:gaussian_circle_results} (H)).

(C-2) Points may be sampled at non-uniform resolution, hence there may  not exist a single good $r$ that can capture all features; see Figure \ref{fig:two_circle_mains}. 
(C-3) Even for a moderate radius $r$, the size of Rips complex becomes large, making persistence computation very costly. 
(C-4) The Rips complex can be a poor approximation of the hidden space when there is background noise; see Figure \ref{fig:two_circle_mains} (F) and (H), where even though the hidden space consists of 5 independent cycles (see Section \ref{sec:exp}), with much background noise, even a small radius $r$ makes the Rips complex connect these noisy points and lose the hidden structure. Removing low-density points can help; however in general that can be challenging when the density distribution is non-uniform. 

\paragraph{A DTM-Rips based approach.} 
The Rips complex is defined based on the Euclidean distance between input points, and does not handle noise or non-uniform point samples well. 
The \emph{distance-to-measure (DTM)} distance is introduced in \cite{CCM11} to provide a more robust way to produce distance field for noisy points. 
We use the work of \cite{BCOS15} to induce a weighted Rips complex from DTM distances, which we now describe briefly. In particular, 
given a set of points $(P, \dd_P)$ equipped with metric $d_P$ (for points $P\subset \reals^d$, $\dd_P$ is the Euclidean distance in $\reals^d$). 
For a fixed integer parameter $k > 0$, let ${\sf kNN}(p)$ denote the set of $k$-nearest neighbor of $p$ in $P$ under metric $\dd_P$. 
For each $p\in P$, we set \emph{(DTM-induced) weight $w_p$} as $w_p = \sqrt{\frac{1}{k} \sum_{q\in {\mathsf kNN}(p)} \dd_P^2(p, q)}$, and the \emph{weighted radius of $p$ at scale $\alpha$} as $r_p(\alpha) =  \sqrt{\alpha^{2} - w_{p}^2}.$
Now given a simplex $\sigma = \{p_{i_0},\ldots, p_{i_s} \}$, we define $\DTMrho(\sigma)$ to be
\begin{multline*}
\DTMrho(\sigma) = \min \{\alpha'  \mid w_{p_{i_j}} \le \alpha', ~\text{and} ~\dd_P(p_{i_{j}}, p_{i_{j'}}) \le r_{p_{i_{j}}}(\alpha') + r_{p_{i_{j'}}}(\alpha'), \forall j\neq j' \in [0,s]\}.
\end{multline*}
This gives an ordering of all possible simplices formed by points in $P$ (again, edges and triangles are needed), and the resulting filtration is called \emph{DTM-Rips filtration $\Fcal_{\DTMrho}$}. 
Equivalently, consider the DTM-weighted Rips complex $\wR^\alpha(P)$ at scale $r$ defined as: $ \wR^r(P) = \{ \sigma= \{p_{i_0},\ldots, p_{i_s} \} \mid \DTMrho(\sigma) \le r \}$. The sequence of $\wR^r(P)$ with increasing scales $r=[0, \infty)$ gives rise to the filtration $\Fcal_\DTMrho$. 
The weight $\omega_p$ is a certain average distance to the $k$NN of $p$ and thus intuitively an inverse density estimator (high density points have low weight). Given two points $p, q\in P$, the edge $\sigma = (p,q)$ has {\bf smaller $\DTMrho(\sigma)$} if $p$ and $q$ has lower weight (thus {\bf higher density}). Simplices spanned by {\bf higher} density points will enter {\bf earlier} into the filtration $\Fcal_{\DTMrho}$. 

\paragraph{Incorporating data sparsification.} 
However, the size of weighted Rips can still be large. To this end, we deploy the sparsified version of DTM-Rips developed in \cite{BCOS15}. The resulting filtration is denoted by \emph{sparse DTM-Rips $\widehat{\Fcal}_\DTMrho(\eps)$} which uses a sparsification parameter $\eps > 0$. See \cite{BCOS15} for details of its construction.
Our final graph skeletonization algorithm for PCDs, denoted by \DMPCD($P, k, \eps, \delta$), consists of only two steps:

(Step 1). Compute the sparse DTM-Rips filtration $\widehat{\Fcal}_\DTMrho(\eps)$ using parameters $k$ (to compute DTM-weights of points) and $\eps$ (for sparsification). 

(Step 2). Apply \newDM($K, \widehat{\Fcal}_\DTMrho(\eps), \delta$) to compute the graph skeleton of $P$, where $K$ is given implicitly as all simplices in $\widehat{\Fcal}_\DTMrho(\eps)$. 

Intuitively, using the DTM-weight alleviates the problem of noisy points (challenge (C-4)), using sparsification addresses the issue of size (challenge (C-3)), while using the entire sparse DTM-Rips filtration allows us to use all radii/scales (instead of a Rips complex at a fixed radius $r$ as in \mybaseline), thereby addressing challenges (C-1) and (C-2). 
Also, while at a larger radius, the filtration will include edges and triangles spanned by far-away points. Theorem \ref{thm:lexOPHC} guarantees that we will output those important loop features using edges that come in as early as possible, i.e., those spanned by higher density points (with smaller $\DTMrho$ values) whenever possible. This allows \DMPCD{} to capture hidden graphs across different scales. See Figure \ref{fig:two_circle_mains}. 

\section{Experimental Results}
\label{sec:exp}
We compare our \DMPCD{} algorithm with the \mybaseline{} algorithm introduced in Section \ref{sec:PCD}, 
and with SOA graph skeletonization algorithms based on Reeb graph \cite{GSBW11} and Mapper \cite{mapper} (referred to as \myReeb{} and \myMapper{} below). (\myMapper{} can produce higher dimensional structures beyond graph skeleton, although often 1D structures are used in practice.) We test on two synthetic point sets and three real datasets. Unless otherwise specified, we use $k=15$ and $\eps = .99$ in our \DMPCD($P, k,\eps,\delta$); while the persistence simplification threshold $\delta$ depends on the point set at hand. For \mybaseline{}, \myReeb{}, and \myMapper{} we report the results of the best parameters we find for them. In particular, a key input for the \myMapper{} algorithm is an appropriate filter function. We tested several standard choices, including distance to base point, eccentricity, density, graph Laplacian eigenfunction and so on, and report the best results found. \emph{For all experiments, all methodologies are run on the original point cloud data, and the figures showing results of higher dimensional data display projections of the results into a lower dimensional space.}
Significantly more results and details are in the Appendix. 

\paragraph{Overview.}  Methods are run on five total datasets - two lower-dimensional (2-D) synthetic datasets, image patches dataset \cite{carlsson08} (8-D), time-delay embedding of traffic sensor datasets \cite{caltrans} (6-D), and Coil-20 \cite{nene96columbia} (16384-D).  Our experiments show that \DMPCD{} is able to extract the true underlying structure of all of these datasets while the other methodologies struggle with noise (image patches and traffic datasets), capturing features at different scales (synthetic and Coil-20 datasets), and having geometrically faithful outputs (synthetic datasets).  Additionally, the size of the filtration used by \DMPCD{} is consistently smaller than that used by \mybaseline{}.

\paragraph{Synthetic datasets.}
We create two synthetic PCDs to illustrate the behavior of our \DMPCD{} algorithm. {\sf Circle} dataset contains a noisy and non-uniform sample around a hidden circle with 2050 points. \emph{\DMPCD{} is able to recover a geometrically faithful hidden circle.  The other methodologies, which require more parameters, also recover a hidden circle, but with a less geometrically faithful structure.  Additionally, \mybaseline{} requires far greater running time for comparable results.} See Figure \ref{fig:gaussian_circle_results}: the output of our method (in (C)) recovers the hidden circle. 
In comparison, outputs of \mybaseline{} algorithm over the Rips complex $\rips^r(P)$ at different radius values are shown in (E) -- (H). 
The total number of simplices involved in our sparsified DTM-Rips filtration is $368,276$. 
The successful \mybaseline{} result (shown in Figure \ref{fig:gaussian_circle_results} (G)) however requires $7,708,243$ simplices, which is about {\bf 20 fold} increase in size. This results in a drastic run-time difference (2.6 seconds vs. 44.8 seconds) between \DMPCD{} and \mybaseline{}.  In general, \DMPCD{} is more efficient than \mybaseline{} because persistence is computed on a much smaller filtration (see Appendix for fully detailed timing results on all datasets).
Also, in general, it is not clear which $r$ to choose for \mybaseline, and if $r$ is too large (e.g., Figure \ref{fig:gaussian_circle_results} (H)), then the output graph  
is geometrically not faithful any more -- this is because long edges are now present in the Rips complex and can appear early in the lower-star filtration in the \oldDM{} algorithm. In contrast, our output (in (C)) takes advantage of the lex-optimality of the algorithm (Theorem \ref{thm:lexOPHC}) and thus always uses "good" edges (small edges from high density regions that enter the filtration early) first. 
The \myReeb{} approach also uses Rips complex at a fixed scale $r$ and thus has similar issues with \mybaseline{}. 
The \myMapper{ approach (Figure \ref{fig:gaussian_circle_results} (J)) correctly captures topology of the space, but misses some geometric details.}

The top row of Figure \ref{fig:two_circle_mains} shows the reconstruction from a set of 300 points non-uniformly sampled from two circles (of different sizes) with background noise. \emph{\DMPCD{} successfully captures both circles, while other methods either fail to capture both circles, or have a topologically correct output that is less geometrically faithful than our method's output.}
Our algorithm scans through all scales in the filtration and captures both loop features. In contrast, both \mybaseline{} and \myReeb{} can capture only one loop. Using a small radius $r$, they can capture the small loop but not the big one. To capture the large loop, they need to use a large radius $r$ (as in Figure \ref{fig:two_circle_mains} (C) and (D)), at which point the small loop is destroyed in the Rips complex. \myMapper{} is able to capture both loops, but again some geometric details are lost (Figure \ref{fig:two_circle_mains} (E)).  See more results in the Appendix.

\begin{figure*}[htbp]
\begin{center}
\includegraphics[width = 0.8\linewidth]{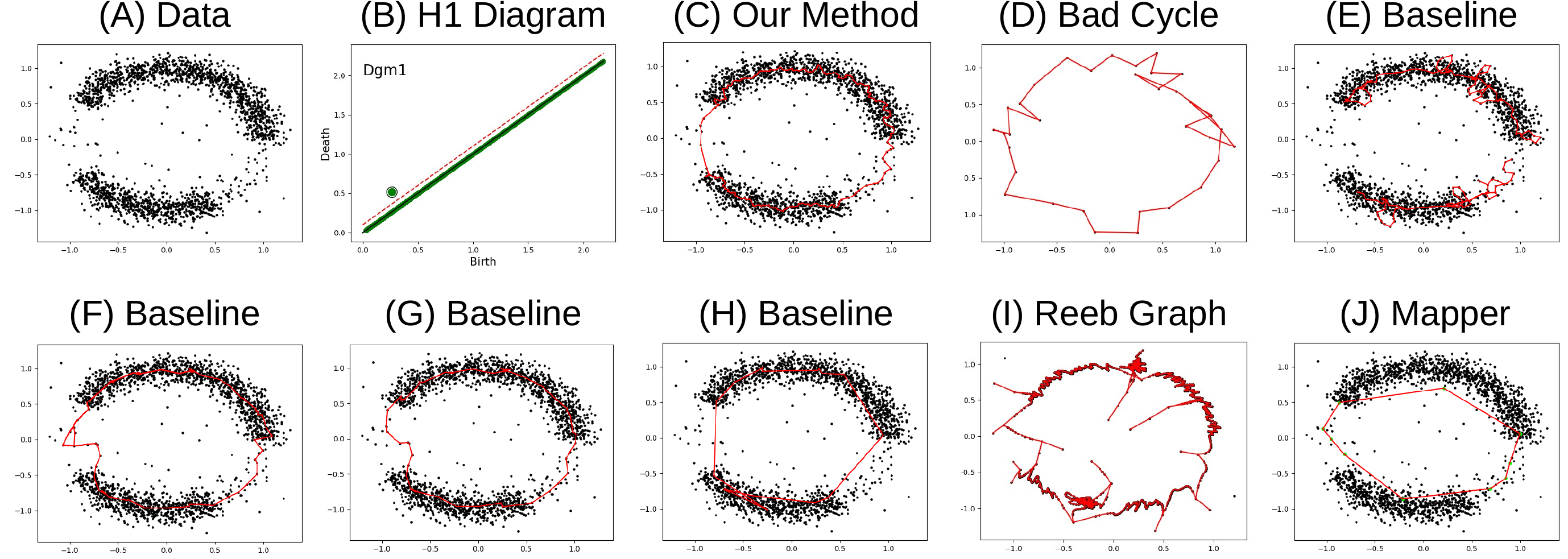}
\end{center}
\vspace*{-0.2in}\caption{{\small (A) Noisy sample of a circle. (B) 1-D persistence diagram w.r.t. our later sparse DTM-Rips filtration. Persistence points are green. Only one point (big point) $\pp$ has persistence larger than some threshold $\delta$ (above the red dotted line). (C) Output of our \DMPCD{} method with persistence threshold $\delta = .25$, which is a lex-optimal persistent cycle w.r.t. the only high persistence point $\pp$ in (B). (D) shows a "bad" persistent cycle w.r.t. the same high persistence point $\pp$. In contrast, our output in (C) uses good (high density) edges whenever possible.  
(E) -- (H) are outputs from the \mybaseline{} algorithm using different radius $r$. (E) $r = 0.1$:
The underlying shape (circle) is not yet captured. (F) $r = 0.2$: There are spurious loops that cannot be simplified via persistence. (G) $r = 0.25$: The circle is recovered; however, the size of $\rips^r(P)$ is now 20 times that of our sparse DTM-Rips filtration.  (H) $r = 1.05$: The output loses geometric details. 
(I) is output from \myReeb{} using $\rips^r(P)$ with radius $r= .25$ and has much noise. 
(J) is output from \myMapper{} using graph Laplacian filter ($k = 15$). }} 
\label{fig:gaussian_circle_results}
\end{figure*}

\paragraph{Image patches dataset.}
The image patches dataset from \cite{carlsson08} contains
 $50K$ points in $\mathbb{S}^7 \subset \reals^8$, each of which corresponds to a 3x3 image patch \cite{Lee2004TheNS}. We subsample $10K$ points randomly so computationally we can experiment with Rips complexes at different radii for the \mybaseline{}.
 \emph{\DMPCD{} is the only method that can extract the true underlying structure from the dataset.  All other methods fail to extract any meaningful structure.}
 The projection of points in 3D (Figure \ref{fig:two_circle_mains} (F)) is very noisy. However, the analysis of \cite{carlsson08} %(after denoising) 
 shows that the underlying space has a "three-circle model", with two circles intersecting the third circle twice but not intersecting each other, thus the first Betti number of the underlying space is $5$.
 Our \DMPCD{} (shown in Figure \ref{fig:two_circle_mains} (G)) successfully recovered the same "three-circle model" (with correct $\beta_1=5$)  directly from raw data {\bf without} preprocessing, and the locations of these (outer, horizontal, and vertical) circles match those shown in \cite{carlsson08}. Both \mybaseline{} and \myMapper{} (in Figure \ref{fig:two_circle_mains} (H) and (I)) fail to capture it. (More details in the Appendix.) Results by \myReeb{} are omitted for this data set, as the algorithm does not handle background noise well and results are poor.

\begin{figure*}[htbp]
\begin{center}
\includegraphics[width = 0.8\linewidth]{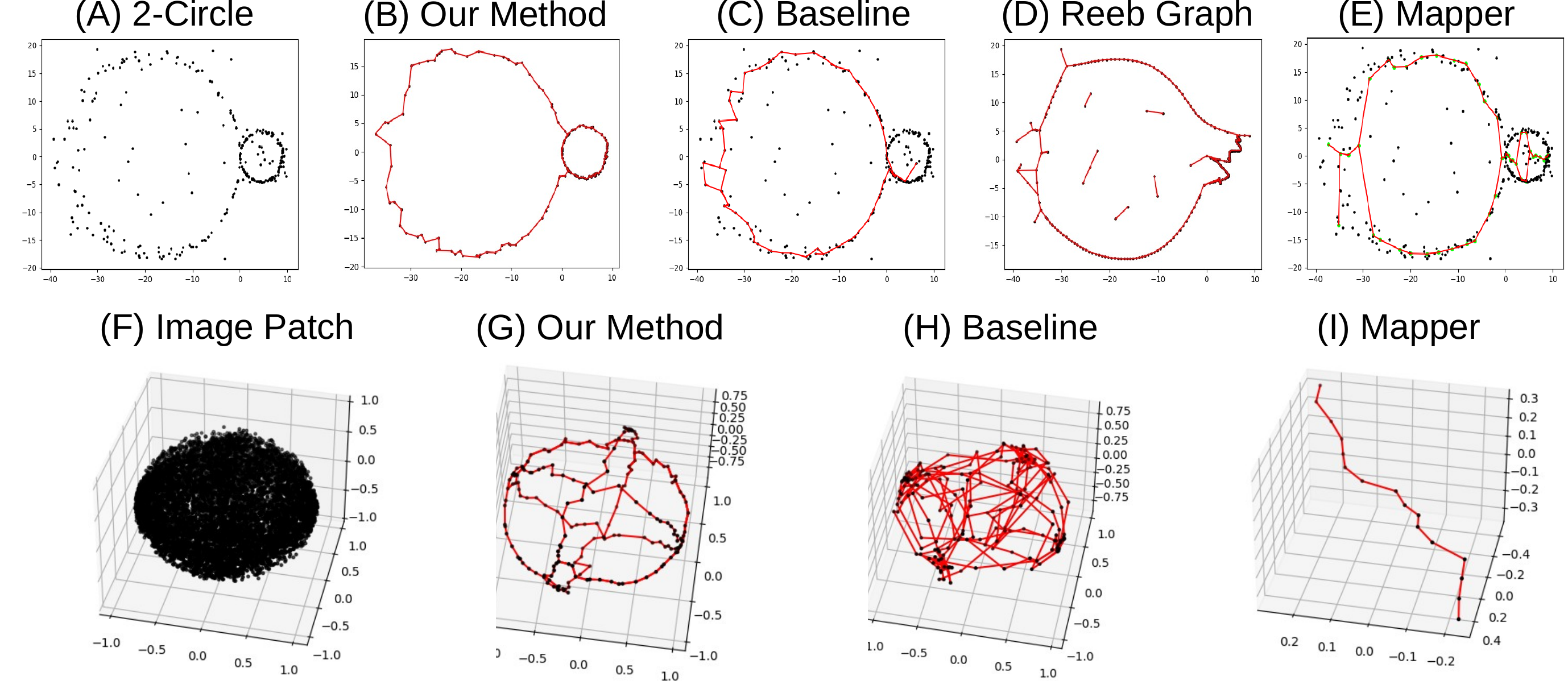}
\end{center}
\vspace*{-0.2in}
\caption{{\small Top row: 2-circle data. Output of our \DMPCD{} method with persistence threshold $\delta = 1.2$ in (B); of \mybaseline{} in (C), of \myReeb{} in (D), and of \myMapper{} in (E). Using a smaller radius $r$ for \mybaseline{} and \myReeb{} will lose the large circle. \myMapper{} output misses geometric details. Bottom row: image patch dataset with projection in $\reals^3$ shown in (F). Our output with persistence threshold $\delta = .146$ in (G) captures the 3-circle model (with $\beta_1 = 5$) perfectly. For \mybaseline{} in (H), further simplification will remove the main circle from the "3-circle model" while keeping all the noisy ones. \myMapper{} (I) (base point filter) is unable to capture any of the loops. }}
\label{fig:two_circle_mains}
\end{figure*} 

\paragraph{Traffic flow dataset.} 
%{\bf Traffic flow dataset.}
We extract two time-series  from \cite{caltrans}, which are the traffic flows 
at detector $\#$409529 from the time-range 10/1/2017 to 10/14/2017 and from the time-range 11/19/2017 to 12/2/2017 (including Thanksgiving). 
Each time-series is mapped to a PCD in $\reals^6$ via time-delay embedding as proposed by \cite{PH15},
who also propose that loops in the resulting PCD can be used to detect quasi-periodic behavior in the original time-series data. We note that a normal time range has one major loop, indicating one major periodicity; while the Thanksgiving period has two: a normal one and one that indicates the traffic pattern over the holiday weekend. \emph{\DMPCD{} recovers these loops much better than \mybaseline{} and \myMapper{}.} Results by \myReeb{} are again omitted due to low quality.

\begin{figure*}[htbp]
\begin{center}
\includegraphics[width = 0.8\linewidth]{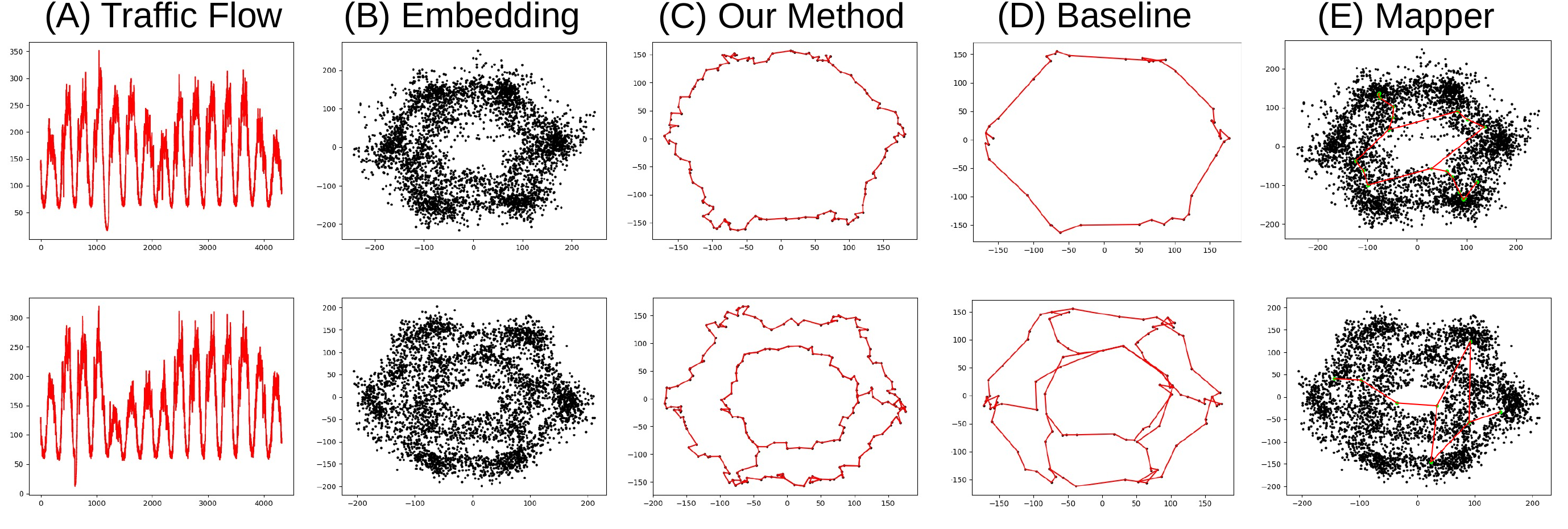}
\end{center}
\vspace*{-0.2in}
\caption{{\small Top row: traffic flow for range (10/1/2017 - 10/14/2017) and bottom row is for range (11/19/2017 - 12/2/2017). (A) input time series, and (B) 2d projections of the time delay embeddings of the time-series.
(C) Outputs of our \DMPCD{} algorithm with persistence thresholds of $\delta = 50$ (top row) and $\delta = 12.5$ (bottom row).  (D) Outputs of the \mybaseline{} approach. 
For the Thanksgiving period, any further simplification will destroy the outer-loop but not the cross connections.
(E) Outputs of \myMapper{} with graph Laplacian filter ($k = 15$). }} 
\label{fig:traffic}
\end{figure*} 

\paragraph{Coil-20 dataset.} 

In our final experiment, we use the Coil-20 dataset provided by \cite{nene96columbia}.  More specifically, we take a subset of 17 objects - removing objects 5, 6, and 19.  Objects 5 and 9 are both medicine boxes, and objects 3, 6, and 19 are toy cars, and we wanted to evaluate our method's performance on a dataset containing unique objects.  We refer to this subset as Coil-17.  Following the process used by \cite{mcinnes2020umap} to convert images to point clouds, we convert each 128 x 128 gray scale image into a 16384 dimensional vector. Hence the input is a set of 1224 points in $\reals^{16384}$. Outputs of other methods can be found in the Appendix.

We visualize the data in two dimensions using UMAP dimensionality reduction with L1 metric.  Presumably, each class forms a high-dimensional loopy shape.  We run \DMPCD{} using L1 metric with $k = 5$.  \emph{\DMPCD{} is able to capture most of the individiual coils UMAP does, while providing a more correct representation of some classes than UMAP.} Shown in Figure \ref{fig:coil-full} is the UMAP reduction with objects uniquely colored (A), the output of our \DMPCD{} algorithm with a persistence threshold of 0 (B), the output of our \DMPCD{} algorithm with a persistence threshold of 56 (C), and the output in (C) after removing the critical edges with L1 length above a threshold of 1700 in the 16384 dimensional original space (D).  The raw output contains the loops that we would expect to see based on our understanding of the data and the shapes formed in the UMAP projection.  The raw output also contains many other edges, revealing more relationships both within individual classes and across multiple classes in the feature space.  We remove the longer edges in order to better highlight the features of the output that capture individual objects.

\begin{figure}[H]
\begin{center}
\includegraphics[width = 0.9\linewidth]{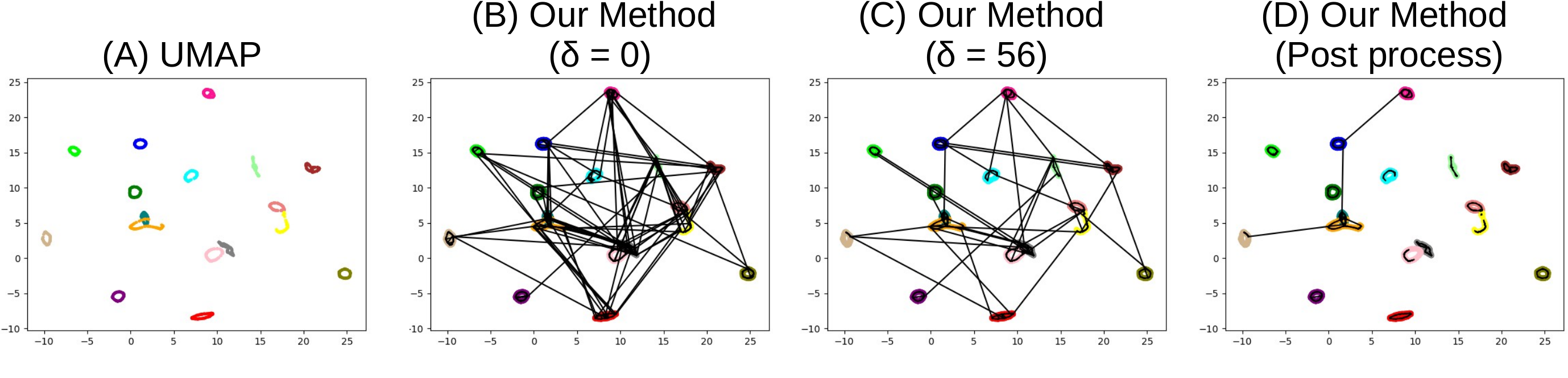}
\end{center}
\vspace*{-0.1in}\caption{{\small Coil: (A) The UMAP projection (using L1 metric) of Coil-17. (B) Output of our \DMPCD{} algorithm with persistence threshold $\delta = 0$.  (C) Output of our \DMPCD{} algorithm with persistence threshold $\delta = 56$. (D) The output shown in (C) with critical edges of L1 length greater than 1700 removed from the output. }}
\label{fig:coil-full}
\end{figure}

Taking a closer look at Figure \ref{fig:coil-full} (D), there are eight objects that the \DMPCD{} captures in the same exact manner that the UMAP projection does.  Close up pictures of these eight objects are shown in Figure \ref{fig:coil-umap-match}.

\begin{figure}[H]
\begin{center}
\includegraphics[width = 0.9\linewidth]{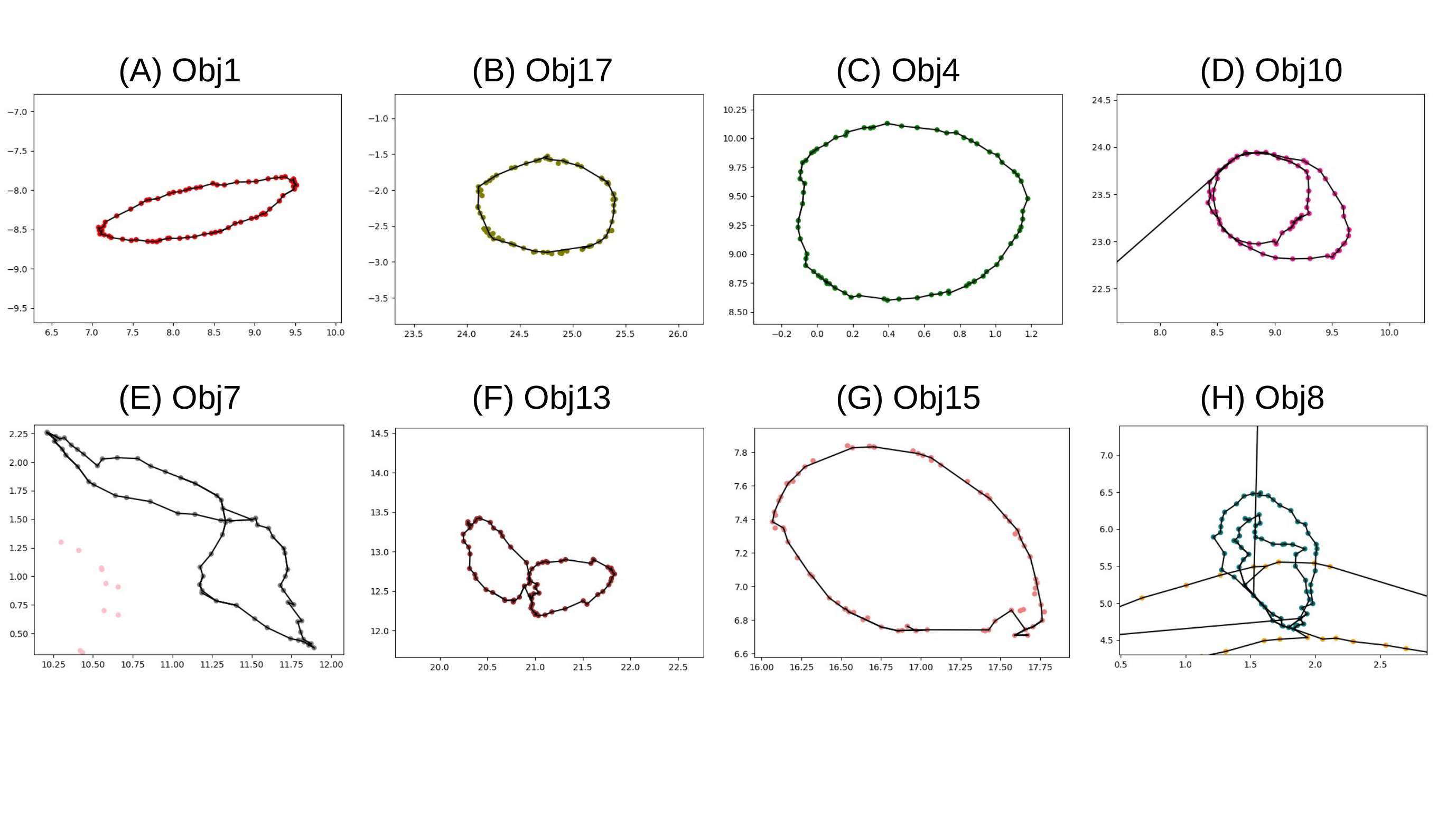}
\end{center}
\vspace*{-0.1in}\caption{{\small Coil: Zoom ins of Figure \ref{fig:coil-full} (D) on the eight objects for which our \DMPCD{} post processed output matches the UMAP projection.}}
\label{fig:coil-umap-match}
\end{figure}

There are also four objects that both the \DMPCD{} output and the UMAP projection capture as loops, but the loops differ between the two methods.  Close up pictures of these four objects are shown in Figure \ref{fig:coil-umap-diff}.  Object 11 (Figure \ref{fig:coil-umap-diff} (A)) is a single loop in the UMAP projection, but is actually two full loops in the \DMPCD{} output.  A closer look (Figure \ref{fig:coil-umap-diff} (B)) at images 16, 17, 54, and 55 shows two separate loops in the output.  The L1 distances between these images in the 16384 dimensional original space (1069.2157046029981, 1038.6980409049952, 693.1607849029981, and 566.901973108997) for pairs (16,54), (17,55), (16,17), and (54,55) respectively) do not match the distances between the pairs in the UMAP projections.  This indicates that the UMAP projection does not preserve the underlying structure of this object, and that the \DMPCD{} output containing two loops is correct.  

A similar result is obtained for Object 14 (Figure \ref{fig:coil-umap-diff} (C)), where UMAP projects a single loop and the \DMPCD{} output contains multiple loops.  Objects 11 and 14 are symmetrical, adding further justification that multiple loops is a better skeletonization.

Object 2 (Figure \ref{fig:coil-umap-diff} (D)) makes a complete loop in the \DMPCD{} output, but the loop looks incomplete in the UMAP projection.  We ran UMAP projections on smaller subsets of Coil-20, some of which project Object 2 as a clear loop (Figure \ref{fig:coil-umap-diff} (E)), whereas the \DMPCD{} output consistently captures Object 2 as a loop.  

Object 20 (Figure \ref{fig:coil-umap-diff} (F)) is captured as the same loop in both the UMAP projection and the \DMPCD{} output, but the \DMPCD{} output has an additional edge dividing the loop.  The edge connects images 44 and 70, which have a L1 distance of 1329.8000237339966 in the original space.  While the other edges adjacent to these nodes are much shorter, other edges that would similarly divide the loop into two are much longer.  For example, the L1 distance between images 19 and 58 is 1822.1608110429997. The dividing edge in the \DMPCD{} output captures this difference, whereas the UMAP projection has no indication of such a difference.

\begin{figure}[H]
\begin{center}
\includegraphics[width = 0.9\linewidth]{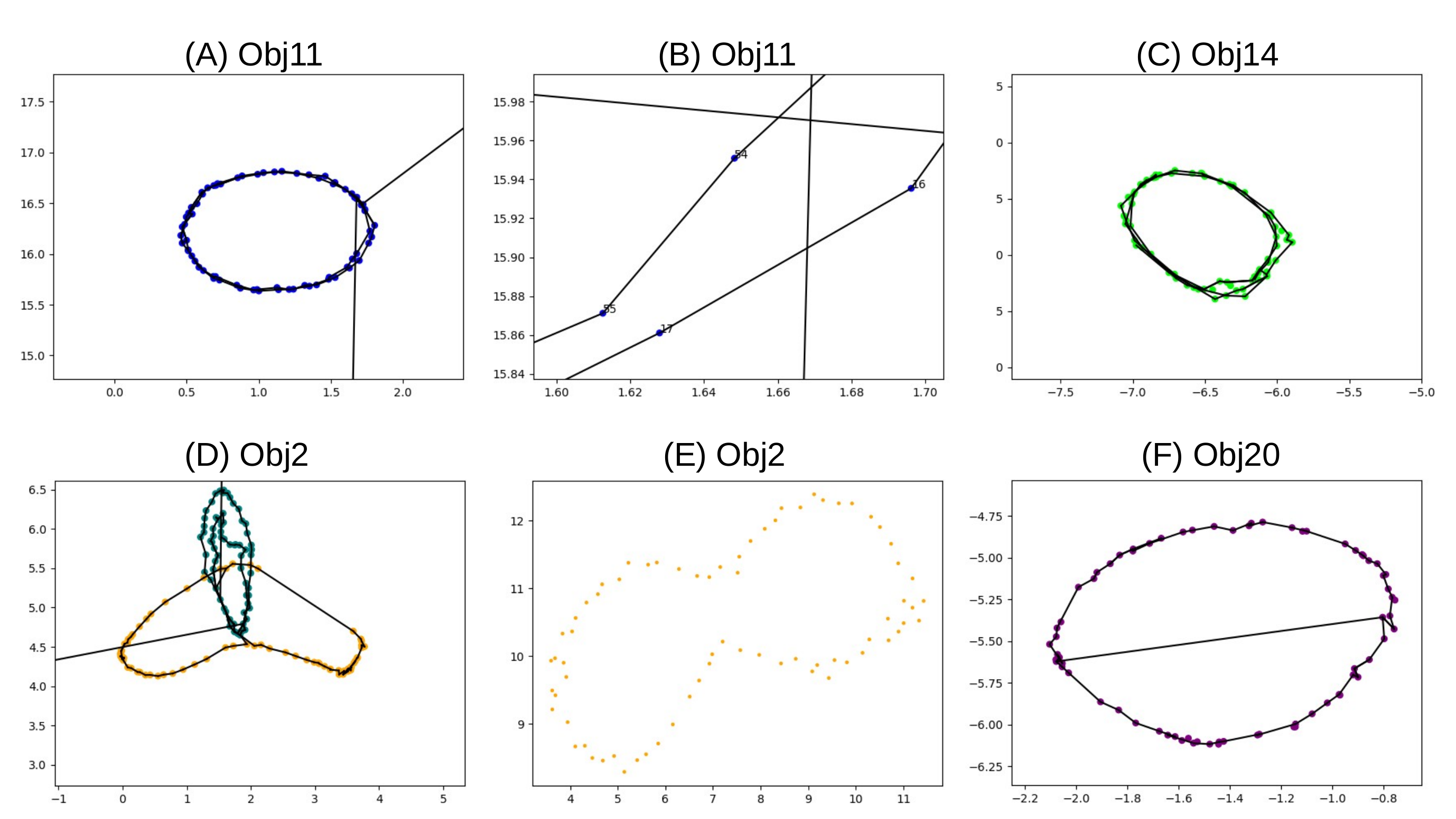}
\end{center}
\vspace*{-0.1in}\caption{{\small Coil: Zoom ins of Figure \ref{fig:coil-full} (D) on the four objects for which our \DMPCD{} post processed output captures a different loop than the UMAP projection. (A) Object 11 - While it appears at first glance that the \DMPCD{} output matches the loop the UMAP projection contains, a closer look reveals (B) that the \DMPCD{} output actually captures two loops.  (C) Object 14 - Similarly to Object 11, the \DMPCD{} output contains multiple loops and UMAP projection only captures one.  (D) Object 2 - UMAP projection is not a complete loop, but \DMPCD{} produces a complete loop.  (E) Object 2 - UMAP projection of only Object 2's images - a complete loop is visible. (F)  Object 20 - UMAP projection shows a single loop, while the \DMPCD{} output captures the same loop, but has an additional edge dividing the loop.}}
\label{fig:coil-umap-diff}
\end{figure}

For the remaining five objects, it is not as clear whether or not the \DMPCD{} output is correct.  Close ups of all five objects are shown in Figure \ref{fig:coil-bad}.  For each object, the figure shows the output at persistence thresholds $\delta = 0$ (first row), $\delta = 56$ (second row), and $\delta = 56$ with critical edges longer than 1700 removed (third row).  Object 3 (Figure \ref{fig:coil-bad} (A)) and Object 18 (Figure \ref{fig:coil-bad} (E)) are not captured as a loop in either the UMAP projection or in any \DMPCD{} output.  The arcs appear to follow the arcs embedded in the UMAP projection.  Object 9 (Figure \ref{fig:coil-bad} (B)) does appear as a loop in the UMAP projection, but is captured as a (double) arc by \DMPCD{}.  Object 12 (Figure \ref{fig:coil-bad} (C)) is captured as a loop in UMAP, but is not in any \DMPCD{} output. However an arc spanning most of the loop is clearly captured.  Under different parameters, \DMPCD{} was able to extract a loop. Object 16 (Figure \ref{fig:coil-bad} (D)) is captured as a loop in UMAP, and a loop is only captured by \DMPCD{} with persistence threshold $\delta = 0$.

\begin{figure}[H]
\begin{center}
\includegraphics[width = 0.9\linewidth]{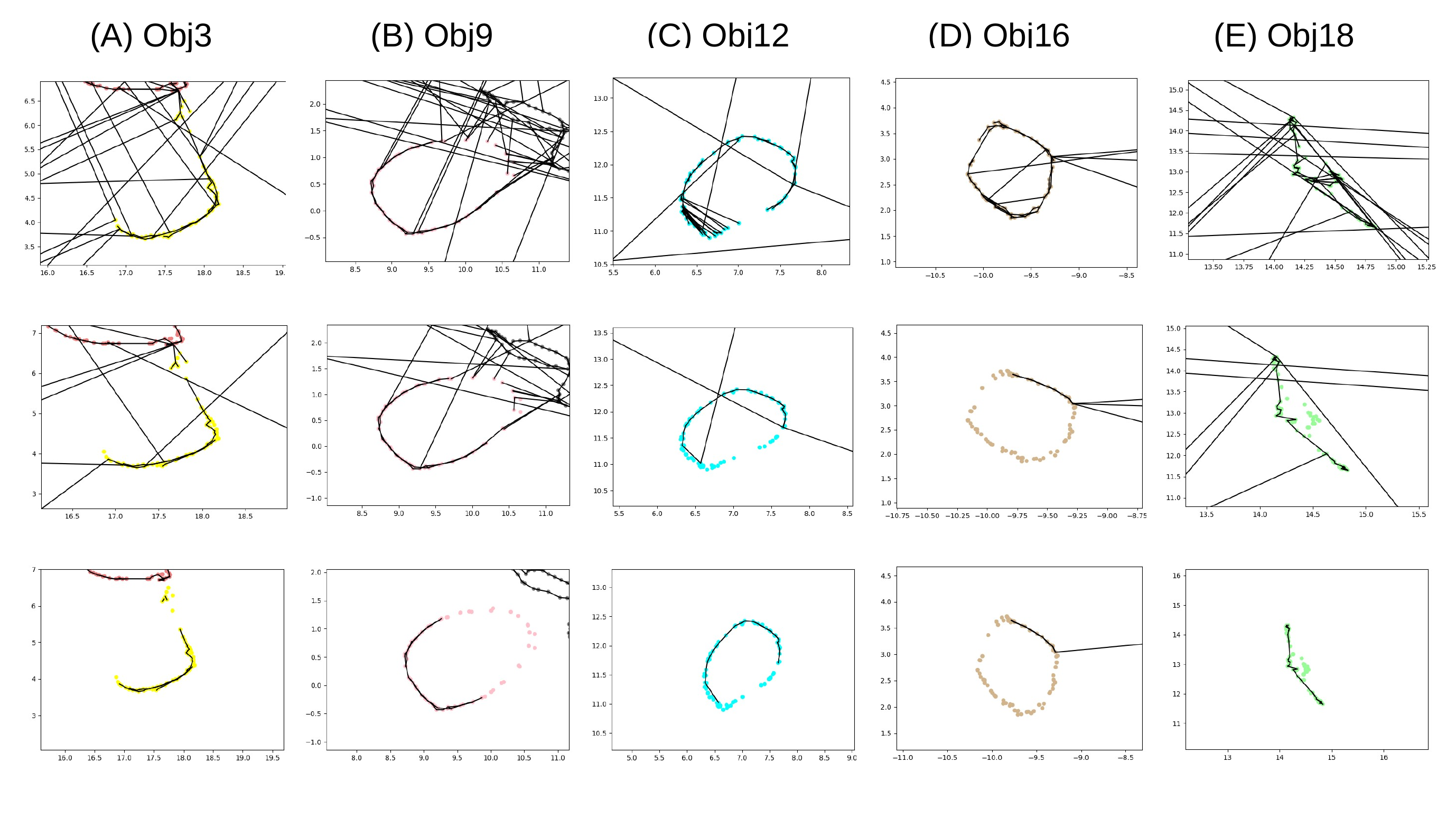}
\end{center}
\vspace*{-0.1in}\caption{{\small Coil: Zoom ins of \DMPCD{} outputs with persistence thresholds $\delta = 0$ (first row), $\delta = 56$ (second row), and $\delta = 56$ with critical edges longer than 1700 removed (third row).  The objects of focus are (A) Object 3, (B) Object 9, (C) Object 12, (D) Object 16, and (E) Object 18. }}
\label{fig:coil-bad}
\end{figure}

A final note on comparing \DMPCD{} to UMAP projections - the metric distortion of UMAP became apparent when viewing the \DMPCD{} outputs.  There is metric distortion within classes - such as Object 20, where there appears to be an extra edge in the \DMPCD{} output because the UMAP embedding does not preserve the distances in the original space.  There is also clear metric distortion in the UMAP embedding with respect to object relationships.  For example, before removing critical edges with length greater than 1700, both Objects 4 and 16 have an edge that connects to Object 8.  However, once the thresholding is applied, the edge connecting Objects 4 and 8 is removed and the edge connecting Objects 16 and 8 remains.  This would indicate that Object 16 is closer to Object 8 than Object 4 is, but Object 4 appears closer in the UMAP embedding.

\section{Concluding Remarks}
We generalized the DM-graph reconstruction algorithm to arbitrary filtrations, proved that the output of this generalized algorithm is meaningful, and developed a method for graph reconstruction from high-dimensional PCDs. Empirical results demonstrate the effectiveness of our \DMPCD{} approach. 

Time complexity is the main limitation of our approach.  The time to compute persistence is a function of the size of the input filtration.  While the theoretical worst case running time is cubic in this size, in practice modern implementations (such as PHAT) perform in subquadratic time (we observe near-linear growth of time w.r.t. the size of simplicial complex in our experiments). Hence reducing the size of filtration is crucial in practice. While the sparsification strategy we used in this paper helps to bring down the size of filtration, the reduction might not be significant enough for very large or more challenging datasets than what we experiment with in our paper. 

In addition to running time causing potential limitations, the raw output of our \DMPCD{} method must be connected.  As shown in Coil-20, with some post-processing we were able to capture individual objects quite easily - but the proper post-processing approach will depend on individual datasets and may not be so straight forward.

\section{Data and Code Availability}

Code for both our new methodology and the baseline approach is publicly available at \url{https://github.com/lucasjmagee/PCD-Graph-Recon-DM}.  The repository also contains all datasets used in this manuscript.

\bibliographystyle{plain}
\bibliography{refs.bib}

\begin{thebibliography}{10}

\bibitem{Aanjaneya11}
M.~Aanjaneya, F.~Chazal, D.~Chen, M.~Glisse, L.~Guibas, and D.~Morozov.
\newblock Metric graph reconstruction from noisy data.
\newblock In {\em Proc. 27th Sympos. Comput. Geom.}, pages 37--46, 2011.

\bibitem{BMWL20}
S.~Banerjee, L.~Magee, D.~Wang, X.~Li, B.~Huo, J.~Jayakumar, K.~Matho, M.~Lin,
  K.~Ram, M.~Sivaprakasam, J.~Huang, Y.~Wang, and P.~Mitra.
\newblock Semantic segmentation of microscopic neuroanatomical data by
  combining topological priors with encoder-decoder deep networks.
\newblock {\em Nature Machine Intelligence}, 2:585--594, 2020.

\bibitem{BN03}
M.~Belkin and P.~Niyogi.
\newblock Laplacian eigenmaps for dimensionality reduction and data
  representation.
\newblock {\em Neural Computation}, 15(6):1373--1396, 2003.

\bibitem{BQWZ12}
Mikhail Belkin, Qichao Que, Yusu Wang, and X.~Zhou.
\newblock Toward understanding complex data: graph laplacians on manifolds with
  singularities and boundaries.
\newblock In {\em Conf. Learning Theory (COLT)}, pages 36.1--36.26, 2012.
\newblock Journal of Machine Learning Research -- Proceedings Track 23.

\bibitem{BCOS15}
Micka\"{e}l Buchet, Fr{\'e}d{\'e}ric Chazal, Steve~Y. Oudot, and Donald~R.
  Sheehy.
\newblock Efficient and robust persistent homology for measures.
\newblock In {\em Proceedings of the Twenty-sixth Annual ACM-SIAM Symposium on
  Discrete Algorithms}, SODA '15, pages 168--180, Philadelphia, PA, USA, 2015.
  Society for Industrial and Applied Mathematics.

\bibitem{cai2021equivariant}
Chen Cai, Nikolaos Vlassis, Lucas Magee, Ran Ma, Zeyu Xiong, Bahador Bahmani,
  Teng-Fong Wong, Yusu Wang, and WaiChing Sun.
\newblock Equivariant geometric learning for digital rock physics: estimating
  formation factor and effective permeability tensors from morse graph, 2021.

\bibitem{caltrans}
{California Department of Transportation}.
\newblock Traffic flows at detector $\#$409529, 2017.

\bibitem{carlsson08}
Gunnar Carlsson, Tigran Ishkhanov, Vin Silva, and Afra Zomorodian.
\newblock On the local behavior of spaces of natural images.
\newblock {\em International Journal of Computer Vision}, 76:1--12, 01 2008.

\bibitem{CCM11}
F.~Chazal, D.~Cohen-Steiner, and Q.~M{\'e}rigot.
\newblock Geometric inference for probability measures.
\newblock {\em Foundations of Computational Mathematics}, 11:733--751, 2011.

\bibitem{CHS15}
Fr{\'e}d{\'e}ric Chazal, Ruqi Huang, and Jian Sun.
\newblock Gromov---hausdorff approximation of filamentary structures using
  reeb-type graphs.
\newblock {\em Discrete Comput. Geom.}, 53(3):621--649, April 2015.

\bibitem{CD18}
Fr{\'e}d{\'e}ric Chazal, Vin~de Silva, Marc Glisse, and Steve Oudot.
\newblock {\em The structure and stability of persistence modules}.
\newblock Springer, 2018.

\bibitem{CLV19}
David Cohen-Steiner, Andr\'{e} Lieutier, and Julien Vuillamy.
\newblock Lexicographic optimal chains and manifold triangulations, 2019.
\newblock available at URL:
  https://hal.archives-ouvertes.fr/hal-02391190/document.

\bibitem{CLV20}
David Cohen-Steiner, Andr\'{e} Lieutier, and Julien Vuillamy.
\newblock Lexicographic optimal homologous chains and applications to point
  cloud triangulations.
\newblock In {\em 36th Sympos. Comput. Geom. (SoCG)}, 2020.
\newblock to appear, see also url:
  https://hal.archives-ouvertes.fr/hal-02391240/document.

\bibitem{DRS15}
O.~Delgado-Friedrichs, V.~Robins, and A.~Sheppard.
\newblock Skeletonization and partitioning of digital images using discrete
  morse theory.
\newblock {\em IEEE Trans. Pattern Anal. Machine Intelligence}, 37(3):654--666,
  March 2015.

\bibitem{DWW18}
T.~K. Dey, J.~Wang, and Y.~Wang.
\newblock Graph reconstruction by discrete morse theory.
\newblock In {\em Proc. Internat. Sympos. Comput. Geom.}, pages 31:1--31:15,
  2018.

\bibitem{DWW19}
Tamal Dey, Jiayuan Wang, and Yusu Wang.
\newblock Road network reconstruction from satellite images with machine
  learning supported by topological methods.
\newblock In {\em Proc. 27th {ACM} {SIGSPATIAL} Intl. Conf. Adv. Geographic
  Information Systems (GIS)}, pages 520--523, 2019.

\bibitem{DeyHM20}
Tamal~K. Dey, Tao Hou, and Sayan Mandal.
\newblock Computing minimal persistent cycles: Polynomial and hard cases.
\newblock In Shuchi Chawla, editor, {\em Proceedings of the 2020 {ACM-SIAM}
  Symposium on Discrete Algorithms, {SODA} 2020, Salt Lake City, UT, USA,
  January 5-8, 2020}, pages 2587--2606. {SIAM}, 2020.

\bibitem{DWW17}
Tamal~K. Dey, Jiayuan Wang, and Yusu Wang.
\newblock Improved road network reconstruction using discrete morse theory.
\newblock In {\em Proc. 25th {ACM} {SIGSPATIAL} Intl. Conf. Adv. Geographic
  Information Systems (GIS)}, pages 58:1--58:4, 2017.

\bibitem{donoho2003hel}
D.L. Donoho and C.~Grimes.
\newblock {Hessian eigenmaps: Locally linear embedding techniques for
  high-dimensional data}.
\newblock {\em Proceedings of the National Academy of Sciences},
  100(10):5591--5596, 2003.

\bibitem{EH10}
Herbert Edelsbrunner and John Harer.
\newblock {\em Computational Topology: {An} Introduction}.
\newblock Amer. Math. Soc., Providence, Rhode Island, 2010.

\bibitem{For98}
R.~Forman.
\newblock Combinatorial vector fields and dynamic systems.
\newblock {\em Mathematische Zeitschrift}, 228(4):629--681, 1998.

\bibitem{For01}
Robin Forman.
\newblock A user's guide to discrete {Morse} theory.
\newblock {\em S\'{e}minare Lotharinen de Combinatore 48}, 2002.

\bibitem{GSBW11}
Xiaoyin Ge, Issam~I. Safa, Mikhail Belkin, and Yusu Wang.
\newblock Data skeletonization via reeb graphs.
\newblock In J.~Shawe-Taylor, R.~S. Zemel, P.~L. Bartlett, F.~Pereira, and
  K.~Q. Weinberger, editors, {\em Advances in Neural Information Processing
  Systems 24}, pages 837--845. Curran Associates, Inc., 2011.

\bibitem{GDN07}
A.~Gyulassy, M.~Duchaineau, V.~Natarajan, V.~Pascucci, E.~Bringa,
  A.~Higginbotham, and B.~Hamann.
\newblock Topologically clean distance fields.
\newblock {\em IEEE Trans. Visualization Computer Graphics}, 13(6):1432--1439,
  Nov 2007.

\bibitem{Hastie84}
T.~J. Hastie.
\newblock {\em Principal curves and surfaces}.
\newblock PhD thesis, stanford university, 1984.

\bibitem{Kegl02}
B.~K\'{e}gl and A.~Krzy\.{z}ak.
\newblock Piecewise linear skeletonization using principal curves.
\newblock {\em IEEE Trans. Pattern Anal. Machine Intell.}, 24:59--74, January
  2002.

\bibitem{LRW14}
Fabrizio Lecci, Alessandro Rinaldo, and Larry Wasserman.
\newblock Statistical analysis of metric graph reconstruction.
\newblock {\em J. Mach. Learn. Res.}, 15(1):3425--3446, January 2014.

\bibitem{Lee2004TheNS}
A.~Lee, K.~Pedersen, and D.~Mumford.
\newblock The nonlinear statistics of high-contrast patches in natural images.
\newblock {\em International Journal of Computer Vision}, 54:83--103, 2004.

\bibitem{mcinnes2020umap}
Leland McInnes, John Healy, and James Melville.
\newblock Umap: Uniform manifold approximation and projection for dimension
  reduction, 2020.

\bibitem{nene96columbia}
Nayar and H.~Murase.
\newblock Columbia object image library: Coil-100.
\newblock Technical Report CUCS-006-96, Department of Computer Science,
  Columbia University, February 1996.

\bibitem{Ozertem11}
U.~Ozertem and D.~Erdogmus.
\newblock Locally defined principal curves and surfaces.
\newblock {\em Journal of Machine Learning Research}, 12:1249--1286, 2011.

\bibitem{PH15}
Jose~A. Perea and John Harer.
\newblock Sliding windows and persistence: An application of topological
  methods to signal analysis.
\newblock {\em Found. Comput. Math. (FoCM)}, 15:799--–838, 2015.
\newblock https://doi.org/10.1007/s10208-014-9206-z.

\bibitem{RWS11}
V.~Robins, P.~J. Wood, and A.~P. Sheppard.
\newblock Theory and algorithms for constructing discrete morse complexes from
  grayscale digital images.
\newblock {\em IEEE Trans. Pattern Anal. Machine Intelligence},
  33(8):1646--1658, Aug 2011.

\bibitem{Roweis2000}
S.T. Roweis and L.K. Saul.
\newblock {Nonlinear Dimensionality Reduction by Locally Linear Embedding}.
\newblock {\em Science}, 290(5500):2323, 2000.

\bibitem{mapper}
Gurjeet Singh, Facundo Memoli, and Gunnar Carlsson.
\newblock {Topological Methods for the Analysis of High Dimensional Data Sets
  and 3D Object Recognition}.
\newblock In M.~Botsch, R.~Pajarola, B.~Chen, and M.~Zwicker, editors, {\em
  Eurographics Symposium on Point-Based Graphics}. The Eurographics
  Association, 2007.

\bibitem{2011MNRAS}
Thierry Sousbie.
\newblock The persistent cosmic web and its filamentary structure – i. theory
  and implementation.
\newblock {\em Monthly Notices of the Royal Astronomical Society}, 414:350 --
  383, 06 2011.

\bibitem{Tenenbaum2000}
J.B. Tenenbaum, V.~Silva, and J.C. Langford.
\newblock {A Global Geometric Framework for Nonlinear Dimensionality
  Reduction}.
\newblock {\em Science}, 290(5500):2319, 2000.

\bibitem{WWL15}
S.~Wang, Y.~Wang, and Y.~Li.
\newblock Efficient map reconstruction and augmentation via topological
  methods.
\newblock In {\em Proc.\ 23rd ACM SIGSPATIAL}, page~25. ACM, 2015.

\bibitem{Wuetal17}
Pengxiang Wu, Chao Chen, Yusu Wang, Shaoting Zhang, Changhe Yuan, Zhen Qian,
  Dimitris~N. Metaxas, and Leon Axel.
\newblock Optimal topological cycles and their application in cardiac
  trabeculae restoration.
\newblock In {\em Information Processing in Medical Imaging - 25th
  International Conference, {IPMI} 2017, Boone, NC, USA, June 25-30, 2017,
  Proceedings}, pages 80--92, 2017.

\end{thebibliography}

\newpage

\appendix

\section{Further Study of Alternative Approaches}
\label{appendix:sec:app}

\myparagraph{Different input triangulations for baseline}
Our main experiments compare the quality of our \DMPCD{} method to the \mybaseline{} algorithm.  \mybaseline{} takes an input triangulation, for which we chose the Rips complex at a fixed radius.  We tested other input triangulations to highlight that the baseline approach fails regardless of the input triangulation.  Results are shown in Figure \ref{fig:baseline_swr}. Even using sparse weighted Rips complex at a fixed radius large enough to capture the larger feature with less noise compared to a regular Rips complex, the points forming the smaller feature are connected by nearly a clique.  Using this triangulation with any valid density function as input for the \mybaseline{} algorithm results in the smaller feature being lost.  It is also shown that using the weighted Rips complex without sparsification results in a similar triangulation and final output.

\myparagraph{Different input filtrations for generalized algorithm}
Our \DMPCD{} algorithm takes a sparse weighted Rips filtration of a point cloud dataset.  However, we generalized the discrete Morse graph reconstruction algorithm to take an arbitrary filtration.  To highlight the utility of the sparse weighted Rips filtration, we run the generalized discrete Morse graph reconstruction algorithm with both the regular Rips filtration and the regular weighted Rips filtration.  Results are shown in Figure \ref{fig:filtrations}.  Using the regular Rips filtration, the output captures the two features with a lot of additional noise.  Trying to use persistence thresholding to remove the noise will remove the smaller feature before all noise is removed.  The regular weighted Rips filtration is able to perfectly capture both features, similarly to using the sparse weighted Rips filtration.  However, because the persistence computation of the filtration is a bottleneck, the sparse filtration is a superior option for our \DMPCD{} algorithm.

\myparagraph{Dimensionality reduction of noisy data}
For noisy datasets, such as the image patches dataset, dimensionality reduction techniques alone fail to reveal meaningful structure.  Results of such techniques are shown in Figure \ref{fig:dim_reduct_10k}.  The main paper shows our \DMPCD{} algorithm extracts a clear three circle structure that is known to be the true underlying structure of the image patches data.  However, PCA, tSNE, and UMAP projections of the image patches dataset reveal no meaningful structure (Figure \ref{fig:dim_reduct_10k} (A) - (C)). This is because these methods do not look to preserve metric relations.  In particular, tSNE attempts to cluster data and UMAP attempts to preserve continuous structure.
For cleaner data, such as Coil-20 (Figure \ref{fig:dim_reduct_10k} (D)), we see that UMAP is able to capture structure.  However, even applying PCA and UMAP (Figure \ref{fig:dim_reduct_10k} (E) and (F)) to the much cleaner $X(15,30)$ subset of image patches, we see that UMAP is still unable to capture the known three circle structure of the data.
Running the \mybaseline{} and \myMapper{} approaches on the PCA reduced image patches data also fails to extract the correct structure.  Results are shown in Figure \ref{fig:dim_reduct_graph_recon}. Running \mybaseline{} with a persistence threshold $\delta = 2$ results in a graph where three circles appear visible (Figure \ref{fig:dim_reduct_graph_recon} (B)).  However, the topology is incorrect, as all circles intersect twice (the first Betti number is equal to 7).  Raising the persistence threshold to 4 (Figure \ref{fig:dim_reduct_graph_recon} (C)) results in an output with the correct first Betti number equal to 5, but we have clearly lost the 3 circles.  \myMapper{} fails on the PCA reduced data and the output is very similar to the output on the original data (Figure \ref{fig:dim_reduct_graph_recon} (D)).  This example highlights a general problem with performing dimensionality reduction then performing graph reconstruction - one needs to reduce to an appropriate dimension.  Clearly it would not be possible to extract the correct graph structure from the images patches dataset if it were first reduced to 2 dimensions. It turns out that reducing to 3 dimensions is also too much, as we are unable to capture the proper (dis)connections between circles.  Not having to reduce dimension, and more so not needing to know the limit for dimensionality reduction, is a huge advantage to our method.

\begin{figure}
\begin{center}
\includegraphics[width = 0.9\linewidth]{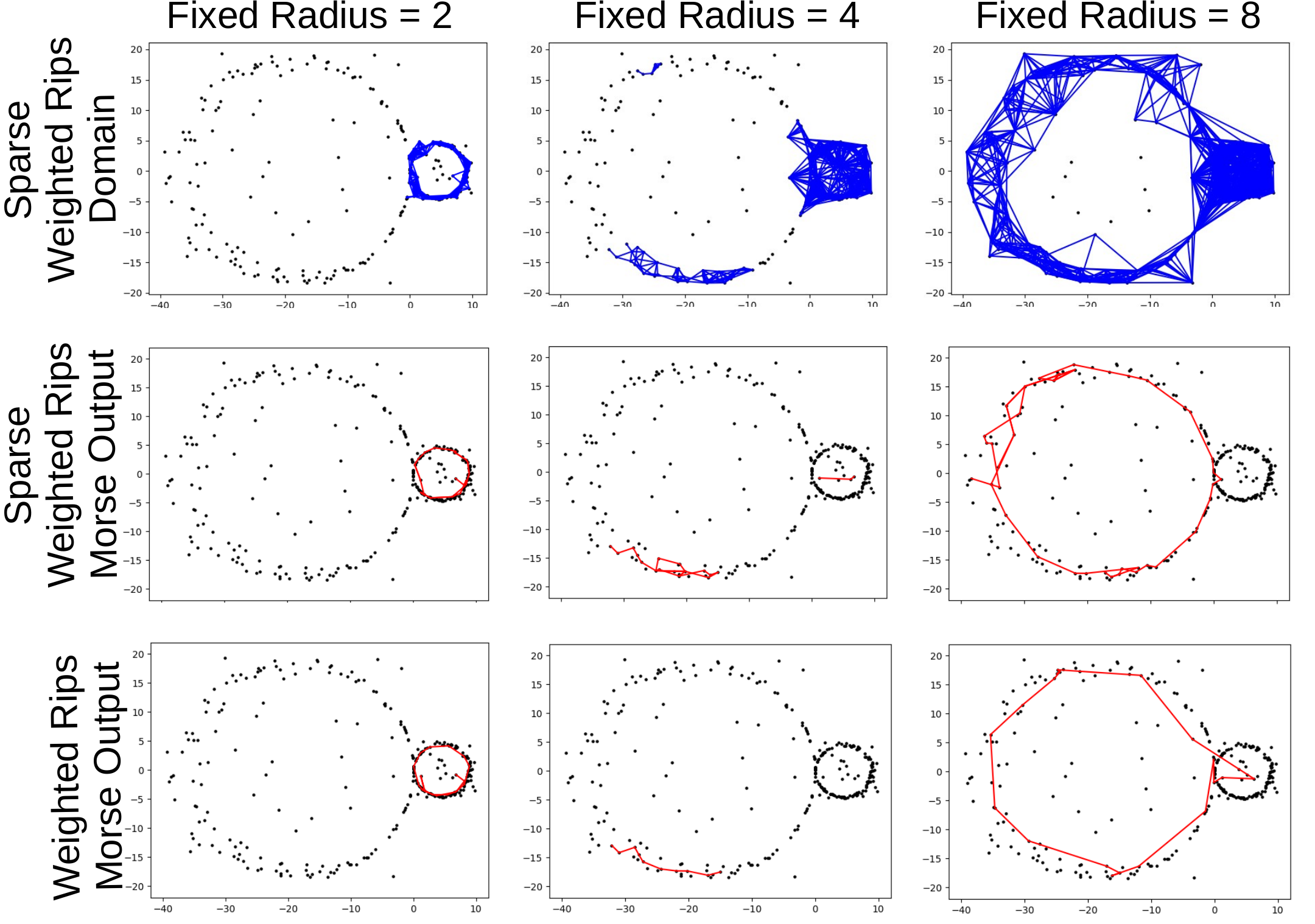}
\end{center}
\vspace*{-0.1in}\caption{{\small Sparse weighted Rips complex at fixed radii (first row) with results of baseline using sparse weighted Rips complex (second row) and full weighted Rips complex (third row).  Fixed radii values of 2 (first column), 4 (second column), and 8 (third column) are shown. }}
\label{fig:baseline_swr}
\end{figure}

\begin{figure}
\begin{center}
\includegraphics[width = 0.9\linewidth]{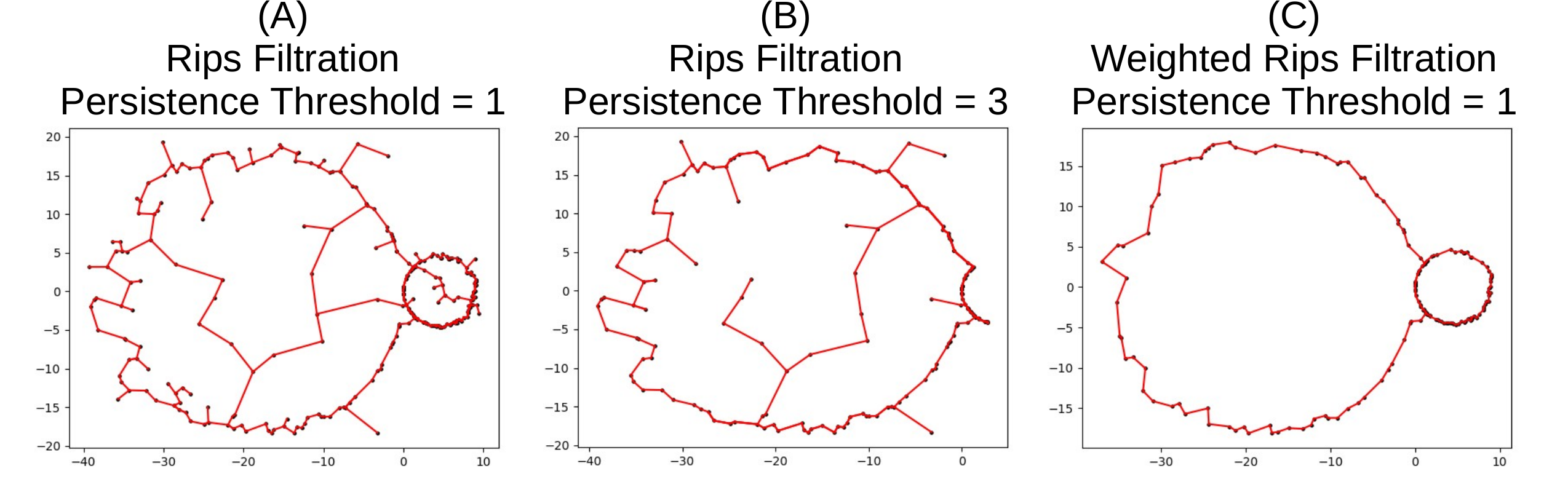}
\end{center}
\vspace*{-0.1in}\caption{{\small Results of using the generalized discrete Morse algorithm with the regular Rips filtration as input.  At a lower persistence threshold (A), both features are captured with additional noise.  At a higher persistence threshold (B), the smaller feature is lost while some noise remains.  Using the weighted Rips filtration as input (C), the algorithm is able to recover both features with no noise.}}
\label{fig:filtrations}
\end{figure}

\begin{figure}
\begin{center}
\includegraphics[width = 0.9\linewidth]{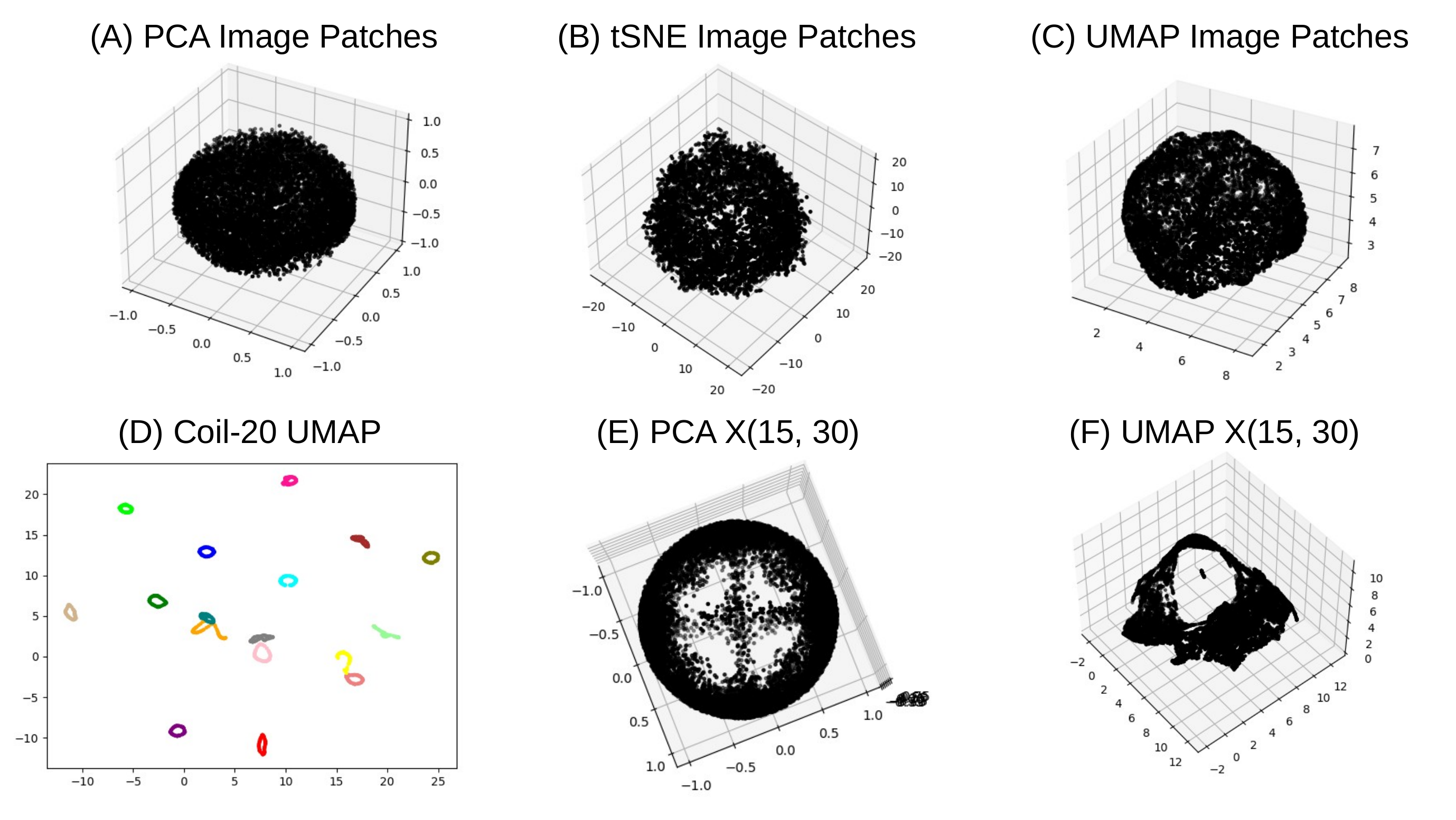}
\end{center}
\vspace*{-0.1in}\caption{{\small Outputs of various dimensionality reduction techniques (PCA, tSNE, and UMAP) performed on the 10,000 image patch subset ((A) - (C)), Coil-20 (D), and $X(15, 30)$ ((E) and (F)). }}
\label{fig:dim_reduct_10k}
\end{figure}

\begin{figure}
\begin{center}
\includegraphics[width = 0.9\linewidth]{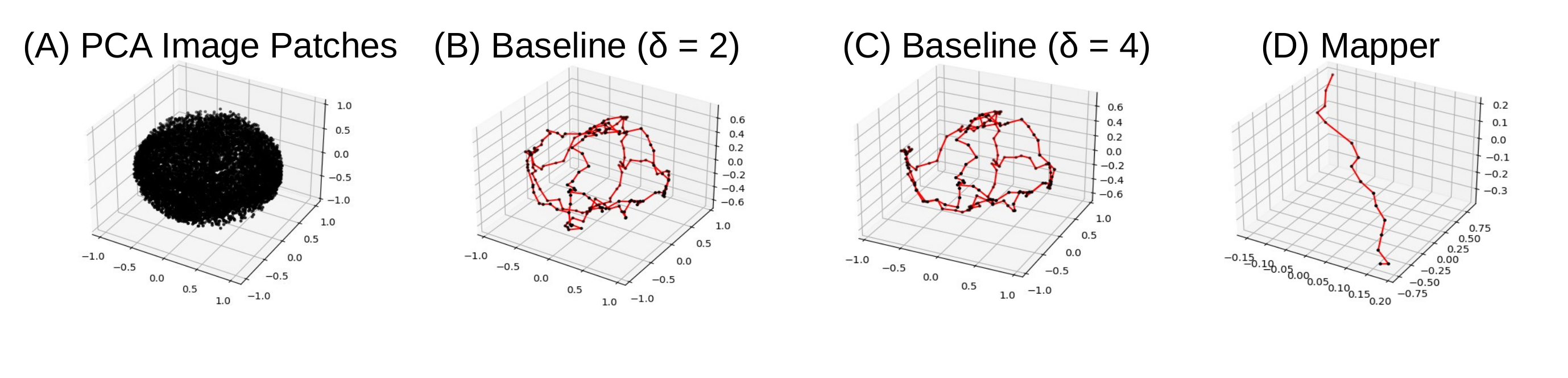}
\end{center}
\vspace*{-0.1in}\caption{{\small PCA reduction of image patches dataset (A) and outputs of \mybaseline{} with persistence thresholds 2 and 4 (B and C), and \myMapper{} (base point filter) (D) on the PCA reduced image patches dataset. }}
\label{fig:dim_reduct_graph_recon}
\end{figure}

\section{More Details on Experiments}
\label{appendix:sec:exp}

\myparagraph{Comparison of methods}
Our experiments compare the quality of outputs and computational efficiency of our \DMPCD{} method with the \mybaseline{} algorithm and the state-of-the-art \myReeb{} algorithm. We also compare the quality of outputs to those of the \myMapper{} algorithm.  We do not include \myMapper{} running times in our comparisons of computational efficiency because it is significantly faster than the other algorithms.

The \myReeb{} algorithm has two outputs - a contracted output, which contains only non-degree two nodes, and an augmented output, which contains edges going through every possible node in the domain.  While the contracted outputs are useful for examining the topology of the output, they do a poor job of preserving the underlying geometry of the output. On the other hand, the augmented outputs are very noisy because every node is included.  For this reason, the authors of the \myReeb{} algorithm smooth outputs.  We smooth the augmented outputs by subsampling the arcs (non-degree two paths), and then perform standard iterative smoothing on the remaining vertices.  An example is shown in Figure \ref{fig:1circle_reeb_sup}.  Unless otherwise noted, the \myReeb{} results displayed in figures are the smoothed augmented outputs.  Ultimately, the quality of the output is now dependant on the smoothing, and we note that different smoothing techniques may result in better quality outputs. However, the topology of the outputs is often incorrect, and in such cases no smoothing can make the output "correct".

The \myMapper{} algorithm traditionally outputs a simplicial complex and was not developed to explicitly extract underlying graph structures from data.  For all of our experiments, we limit the \myMapper{} output to be a graph (1-dimensional simplicial complex).  Each node in a graph outputted by \myMapper{} represents a cluster computed within the algorithm.  We assign the coordinates of a node to be the average coordinates of the cluster it represents.

\begin{figure}
\begin{center}
\includegraphics[width = 0.9\linewidth]{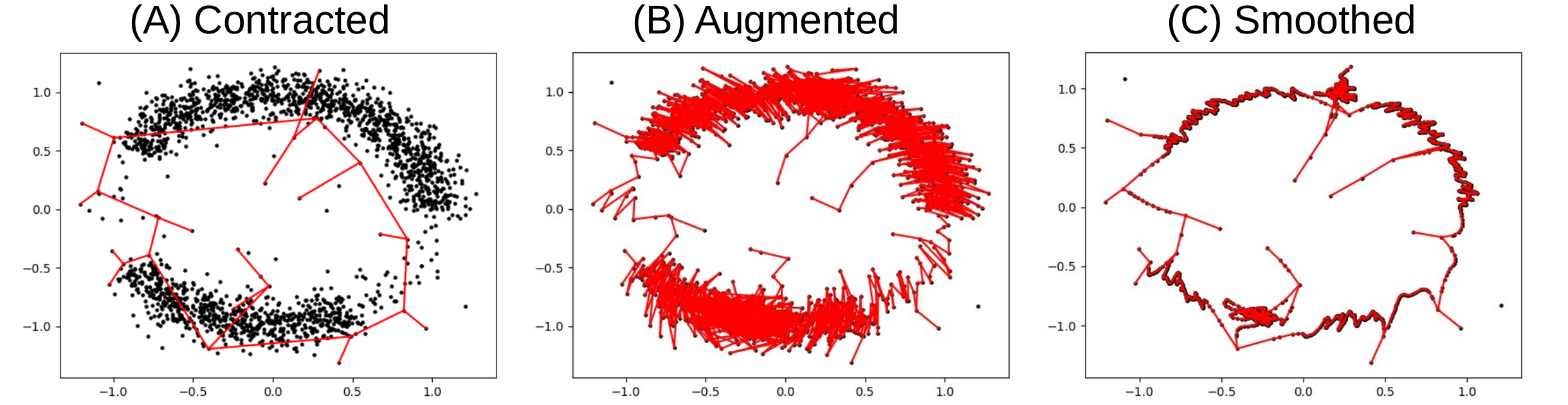}
\end{center}
\vspace*{-0.1in}\caption{{\small \myReeb{} outputs on one circle dataset.  The contracted output (A) and the augmented output (B) of the \myReeb{} algorithm with $r = .25$.  The contracted output is useful in examining the topology of the output, but does a poor job of preserving the geometry of the underlying skeleton.  The augmented output better preserves the geometry but contains a lot of noise. (C) is a smoothed augmented output with less noise.}}
\label{fig:1circle_reeb_sup}
\end{figure}

In our time comparisons, \myReeb{} is much slower than both \DMPCD{} and \mybaseline{}.  While we are using an old implementation from 2011 that may not be optimized, it is known that \myReeb{} is theoretically faster than both \DMPCD{} and \mybaseline{}, which have persistence computation as a bottleneck.  \DMPCD{} tends to be more efficient than \mybaseline{}, as the sparsification in our algorithm builds a filtration that is linear in size with respect to the number of points, whereas the regular Rips complex used in \mybaseline{} results in a filtration of size $O(n^3)$, and $r$ values large enough to capture the underlying skeleton will have much bigger filtrations.

\myparagraph{One Circle dataset.}
The main paper shows that our \DMPCD{} algorithm is able to successfully capture the circle, and both the \mybaseline{} and \myReeb{} capture the circle with $r = .25$.  However, the quality of output for both \mybaseline{} and \myReeb{} is heavily dependent on the value of $r$ - more specifically the corresponding $\rips^r(P)$ complex.  Shown in the first row of Figure \ref{fig:1circle_rips_sup} is the $\rips^r(P)$ complex for $r$ values of .1 (A), .2 (B), .25 (C), and 1.05 (D).  The second and third rows contain results of \mybaseline{} and \myReeb{}.  All \myReeb{} outputs are smoothed with no subsampling, a neighborhood radius of 2 neighbors, and 5 iterations - except for (D), where the output is a spanning tree and smoothing does not improve output quality.  With an $r$ value too small (.1), the underlying skeleton is not contained in $\rips^r(P)$, and neither method will be able to produce a desirable output.  It is not enough to select an $r$ value that results in the complex containing the underlying skeleton.  For $r = .2$, the circle is captured by the $\rips^r(P)$ complex, but so is an additional spurious loop.  Neither \mybaseline{} or \myReeb{} can produce an output not containing the spurious loop.  For $r = 1.05$, there are no additional spurious loops in the $\rips^r(P)$ complex, but \myReeb{} produces a spanning tree and \mybaseline{}, while producing a single loop, loses the geometry of the underlying skeleton.

For this dataset, $r = .25$ was an appropriate selection for both \mybaseline{} and \myReeb.  However, as shown in Table \ref{tab:1circle_times}, the size of the $\rips^r(P)$ complex is much bigger than the sparsified weighted Rips complex used in our \DMPCD{} algorithm.  Complex size is particularly costly for our \DMPCD{} method and the \mybaseline{} method because of the persistence computation.  While \mybaseline{} was able to produce a reasonable output at $r = .25$, it took significantly more time than our \DMPCD{} algorithm.

While smoothing certainly decreases the noise in the \myReeb{} output, the output quality is still worse than that of both \DMPCD{} and \mybaseline{}.  We comment that a different smoothing method may result in a better quality output.

Additionally, the main paper shows that the \myMapper{} approach is also able to successfully capture the circle. We show the $\myMapper{}$ results with a variety of filter functions in Figure \ref{fig:1circle_mapper_sup}.  The graph Laplacian filter and the distance to base point filter are able to capture the circle, while the eccentricity filter and the Gaussian density filter are unable to capture the true underlying structure of the data. The heat maps of the filter functions shown in the first row of Figure \ref{fig:1circle_mapper_sup} provide intuition on why each filters is either successfully or unsuccessfully used to extract the underlying structure with \myMapper{}.  These two filter functions will be the top choices for most of the remaining datasets.

\begin{figure}
\begin{center}
\includegraphics[width = 0.9\linewidth]{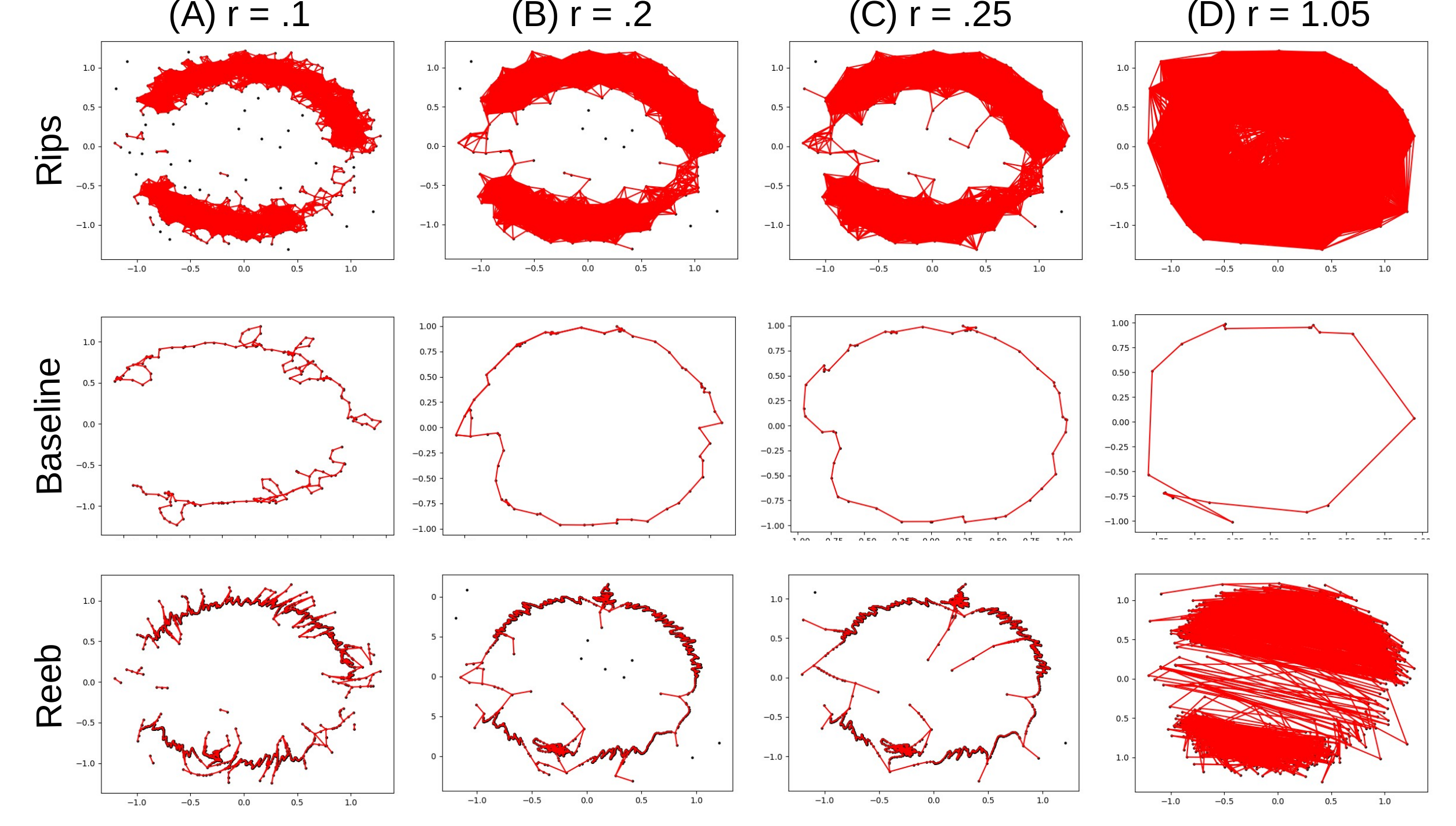}
\end{center}
\vspace*{-0.1in}\caption{{\small One circle data - $\rips^r(P)$ complex (first row), \mybaseline{} outputs (second row), and \myReeb{} outputs (third row). (A) $r = .1$ - $\rips^r(P)$ complex fails to capture the circle, resulting in both methods failing to capture the circle.  (B) $r = .2$ - $\rips^r(P)$ complex now contains the circle, but also contains a spurious loop. Both the \mybaseline{} output, which was generated with persistence threshold $\delta = \infty$, and the \myReeb{} output must contain this spurious loop  (C) $r = .25$ - $\rips^r(P)$ complex now contains the circle without any additional spurious loops.  Both the \mybaseline{} output ($\delta = \infty$) and the \myReeb{} output capture the loop.  (D) $r = 1.05$ - $\rips^r(P)$ complex still contains the circle, as well as many more simplices. As a result, the \mybaseline{} output has lost its nice geometry, with long edges going through high density regions, and the \myReeb{} output is a spanning tree.  }}
\label{fig:1circle_rips_sup}
\end{figure}

\begin{figure}
\begin{center}
\includegraphics[width = 0.9\linewidth]{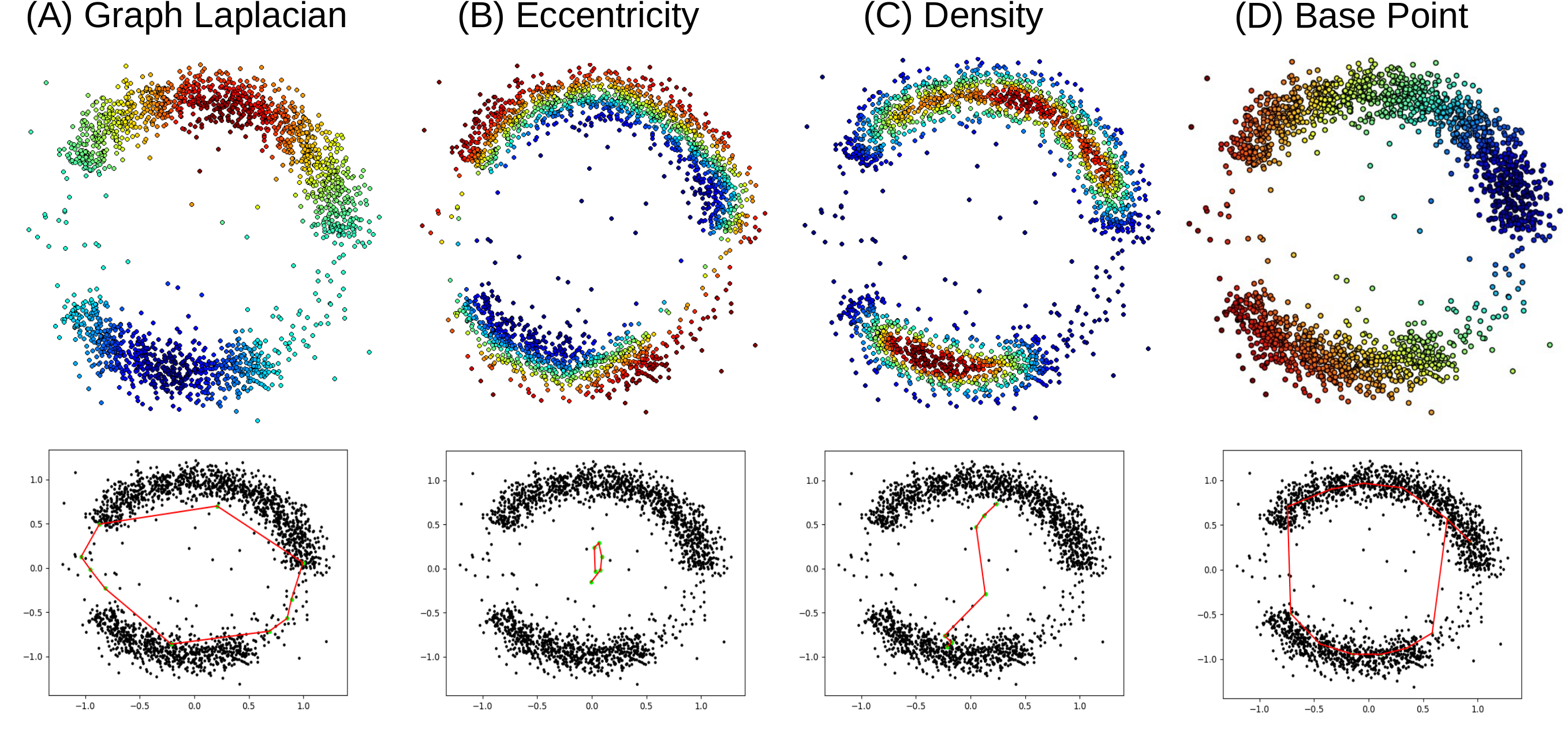}
\end{center}
\vspace*{-0.1in}\caption{{\small One circle data - Filter function values (first row) and corresponding Mapper outputs (second row) with filter functions - (A) graph Laplacian filter ($k = 15$), (B) eccentricity filter, (C) Gaussian density filter, (D) distance to base point filter.}}
\label{fig:1circle_mapper_sup}
\end{figure}

\begin{table}[!htbp]
    \centering
    \begin{tabular}{llll}
        \hline
        \textbf{Method} & \textbf{Radius} & \textbf{$\#$ Simplices} & \textbf{Time (seconds)} \\
        \hline\hline
        {\bf Our Method}  & $\infty$ & 368276 & 2.6 \\
        \hline
        Baseline & .05 & 33368 & .1 \\
        \hline
        Reeb Graph & .05 & 33368 & .03 \\
        \hline
        Baseline & .10 & 356925 & 1.1 \\
        \hline
        Reeb Graph & .10 & 356925 &  2.49\\
        \hline
        Baseline & .15 & 1490149 & 5.4 \\ 
        \hline
        Reeb Graph & .15 & 1490149 & 21.05 \\ 
        \hline
        Baseline & .2 & 3869507 & 13.0 \\ 
        \hline
        Reeb Graph & .2 & 3869507 & 76.46 \\ 
        \hline
        Baseline & .25 & 7708243 & 44.8 \\
        \hline
        Reeb Graph & .25 & 7708243 & 221.25 \\
        \hline
        Baseline & .5 & 43392850 & 231.1\\ 
        \hline
        Reeb Graph & .5 & 43392850 & 3445.10 \\ 
        \hline\hline

    \end{tabular}
    \caption{One circle dataset: Comparison of radius used, $\#$ simplices, and running time of \DMPCD{}, \mybaseline{}, and \myReeb{}. Our algorithm has radius $\infty$ as we run on the full sparse DTM-Rips filtration. } 
    \label{tab:1circle_times}
\end{table}{}

\myparagraph{Two Circle dataset.}
The main paper shows that our \DMPCD{} algorithm is able to successfully capture both circles, while both \mybaseline{} and \myReeb{} failed to capture both circles.  Again, this is because both methods are heavily dependent on the input triangulation (the $\rips^r(P)$ complex).  This complex at various values of $r$ is shown in the first row of Figure \ref{fig:2circle_sup}, while the corresponding \mybaseline{} and \myReeb{} outputs are shown in the second and third rows respectively.  All \myReeb{} outputs are smoothed with no subsampling, a neighborhood radius of 2 neighbors, and 10 iterations.  Neither result can contain the larger circle if the input triangulation itself does not contain the larger circle, so we increase values of $r$ until the complex contains the larger circle.  $r = 1$ is too small to capture even the smaller circle.  At $r = 1.5$, the complex does contain the smaller circle, and both \mybaseline{} and \myReeb{} are able to successfully extract the loop.  However, at $r = 2$ and $r = 3$, the complex still does not contain the larger circle, and more noise around the smaller circle is added to the outputs.  Finally, at $r = 4$, the larger circle is contained within the complex.  However there are two issues.  Firstly, there is a spurious loop in the complex along the larger circle, so while $r = 4$ is able to capture the larger circle, it is still not an appropriate value of $r$. We would need to try to find a new $r$ value that better captures the data if not for the second issue - the smaller circle is lost in both outputs - meaning an appropriate value of $r$ does not exist for either method.  We can see that in the $\rips^r(P)$ complex at $r = 4$, the points forming the smaller circle now nearly form a clique, which results in both \mybaseline{} and \myReeb{} outputs losing the smaller circle.  We conclude that there is no value for $r$ that will result in either method capturing both circles. Running time and simplicial complex size comparisons are shown in Table \ref{tab:2circle_times}.  For radius $r = 4$, we see that the number of simplices used in both \mybaseline{} and \myReeb{} is nearly double that of the filtration used by \DMPCD{}.  As a result, the running times of \mybaseline{} and \myReeb{} are longer than that of \DMPCD{}.

Additionally, the main paper shows that \myMapper{} was able to successfully capture both features of the two circle dataset.  Further results for different filter functions are shown in Figure \ref{fig:2circle_mapper}.  Similarly to the results of the one circle dataset, \myMapper{} was able to successfully capture both features when using either the graph Laplacian filter or distance to base point filter.  Looking at the heat map for the eccentricity filter, we see that it would also appear to be an acceptable choice for this particular dataset.  The output captures the larger feature and is unable to capture the smaller feature.  The density filter once again fails to extract any meaningful structure from the dataset.

\begin{figure}
\begin{center}
\includegraphics[width = 0.9\linewidth]{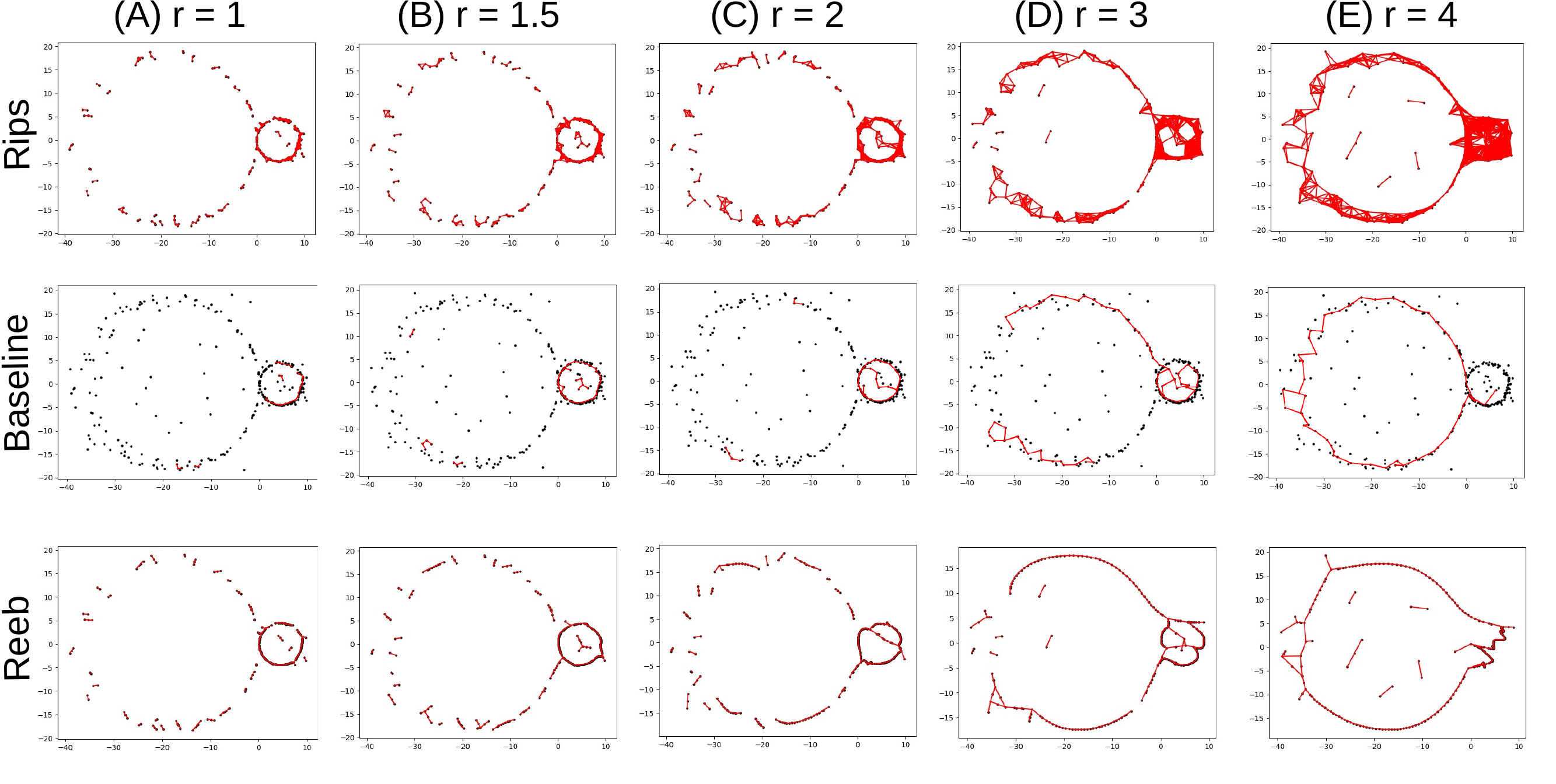}
\end{center}
\vspace*{-0.1in}\caption{{\small The $\rips^r(P)$ complex (first row), \mybaseline{} outputs (second row), and \myReeb{} outputs (third row). All \mybaseline{} outputs were generated using persistence threshold zero, meaning that no simplification occurred.  (A) $r = 1$ - $\rips^r(P)$ complex fails to capture either circle, resulting in both methods failing to capture either circle.  (B) $r = 1.5$ - $\rips^r(P)$ complex now contains the smaller circle but not the larger circle. Both outputs successfully capture the smaller circle, but fail to capture the larger circle.  (C) $r = 2$ - $\rips^r(P)$ complex now connects the noise inside of the smaller circle to the smaller circle, while still not containing the larger circle.  The outputs now capture the smaller circle and some noise inside of the circle, and still fail to capture the larger circle.  (D) $r = 3$ - $\rips^r(P)$ complex still does not contain the larger circle, and contains many edges cutting across the smaller circle.  The \myReeb{} output captures the smaller circle with more noise, while the \mybaseline{} output has begun to lose the smaller circle.  Both outputs fail to capture the larger circle.  (E) $r = 4$ - $\rips^r(P)$ complex now contains the larger circle, as well as a spurious loop, and the points forming the smaller circle now nearly form a clique.  The outputs capture the larger circle, but contain a spurious loop, and the smaller circle is completely lost. }}
\label{fig:2circle_sup}
\end{figure}

\begin{figure}
\begin{center}
\includegraphics[width = 0.9\linewidth]{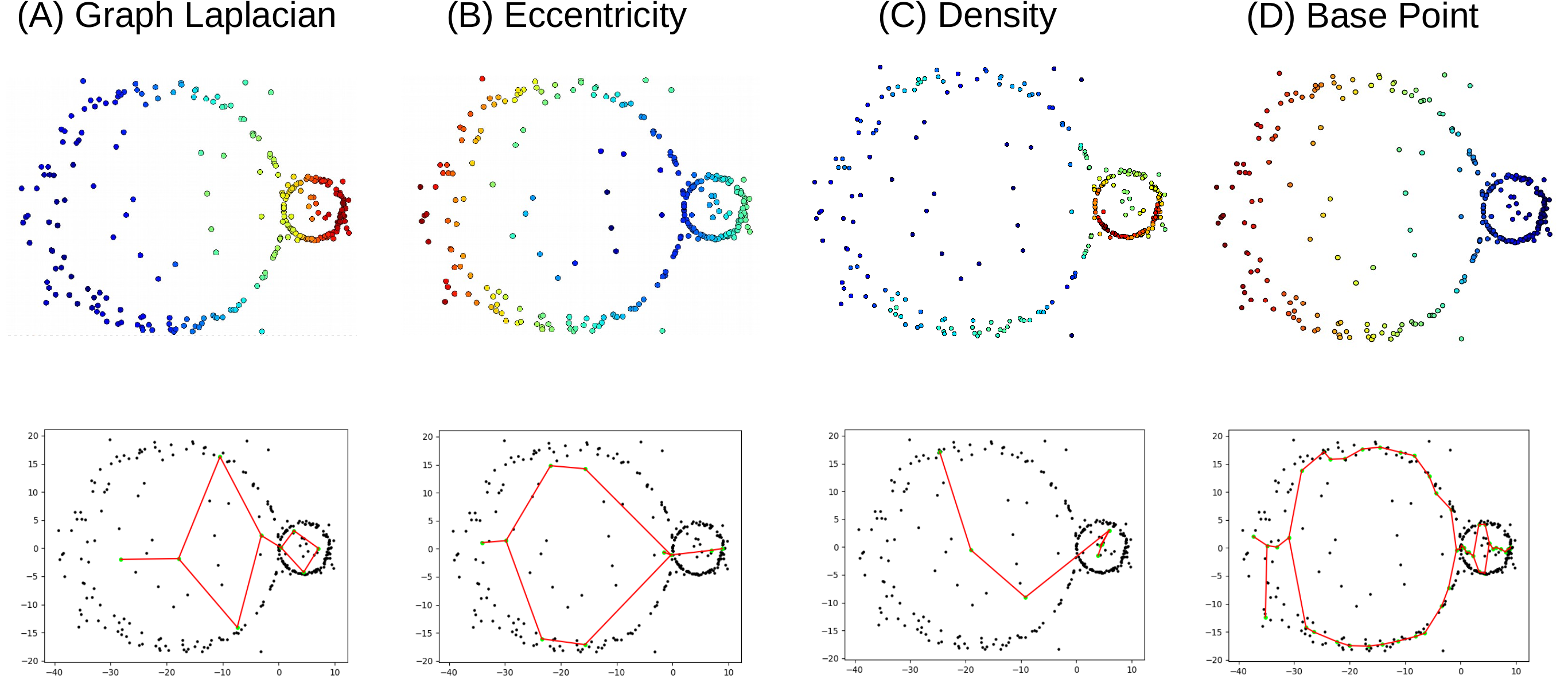}
\end{center}
\vspace*{-0.1in}\caption{{\small \small Two circle dataset: filter function values (first row) and corresponding Mapper outputs (second row) using (A) graph Laplacian filter ($k = 15$), (B) eccentricity filter, (C) Gaussian density filter, (D) distance to base point filter.}}
\label{fig:2circle_mapper}
\end{figure}

\begin{table}[!htbp]
    \centering
    \begin{tabular}{llll}
        \hline
        \textbf{Method} & \textbf{Radius} & \textbf{$\#$ Simplices} & \textbf{Time (seconds)} \\
        \hline\hline
        %\multicolumn{4}{c}{\textbf{\textit{Efficiency Comparison Across Methods}}}\\ \hline
        {\bf Our Method}  & $\infty$ & 19497 & .06 \\
        \hline
        Baseline & 1 & 2182 & .002 \\
        \hline
        Reeb Graph & 1 & 2182 & .01 \\
        \hline
        Baseline & 2 & 8350 & .013 \\
        \hline
        Reeb Graph & 2 & 8350 & .02 \\
        \hline
        Baseline & 3 & 20005 & .045 \\
        \hline
        Reeb Graph & 3 & 20005 & .05 \\
        \hline
        Baseline & 4 & 41349 & .13 \\
        \hline
        Reeb Graph & 4 & 41349 & .25 \\
        \hline\hline

    \end{tabular}
    \caption{Two circle dataset: Comparison of running time of \DMPCD{}, \mybaseline{}, and \myReeb{}. Our algorithm has radius $\infty$ as we run on the full sparse DTM-Rips filtration. } 
    \label{tab:2circle_times}
\end{table}{}

\myparagraph{Image patches dataset.}
The main paper shows that our \DMPCD{} algorithm is able to successfully extract the "three-circle model" from a random 10,000 point subset of the image patches dataset from \cite{carlsson08}, while \mybaseline{}, \myReeb{}, and \myMapper{} methods are unable to do so.  We run \mybaseline{} with $\rips^{.75}(P)$ as the input complex.  We tried several $r$ values less than $.75$, as well as $r = .8$.  For $r$ values less than $.75$, there were many spurious loops that could not be removed with persistence thresholding.  For $r = .8$, the desired three-circles are not completely recovered even with no persistence thresholding. Results at various persistence thresholds are shown in Figure \ref{fig:3circle_baseline}.  Although the output does contain the three circles we wish to extract, it is also made up of several additional loops.  Raising the persistence threshold to 5 removes some of the additional loops, but raising the persistence threshold to 10 removes part of the desired three circle model without removing the remaining additional loops.  In fact, raising the persistence threshold to $\infty$, we see that some of these incorrect loops are a product of the input triangulation, and it is not possible to achieve a desired output from \mybaseline{} with $r = .75$.  While it may still be possible for a "good" $r$ value to exist, it is extremely expensive to compute persistence on triangulations with this many simplices.

Running time and simplicial complex size comparisons are shown in Table \ref{tab:3circle_times}.  For radius $r = .75$, we see that the number of simplices used in \mybaseline{} is over 50,000,000 greater than the number of simplices in the filtration used by \DMPCD{}.  Although \DMPCD{} takes longer to compute persistence even with a smaller filtration, the \DMPCD{} filtration has an implied $r = \infty$, and that any sizable increase to $r = .75$ for \mybaseline{} will result in a significant increase in running time. We note that for all values of $r$, the number of simplices used in \mybaseline{} would be the same number of simplices used by \myReeb{}.

Additionally, the main paper shows that \myMapper{} was unable to capture the true underlying structure of the image patches dataset.  Further results for different filter functions are shown in Figure \ref{fig:3circle_mapper}.  Gaussian density and eccentricity filters fail, as seen in previous datasets.  However, unlike the previous dataset, the graph Laplacian and distance to base point filters also fail to capture the underlying structure.  This data is simply too noisy for \myMapper{} to extract the underlying structure.

Finally, while our algorithm is deterministic, this dataset is generated from a random 10K point subset.  In an attempt to quantify the error,  we generated 10 different random subsets to apply our method to.  On all 10 datasets, our method extracts the 3 circles correctly.  To quantify error, we computed the distance between two output graphs $G_i$, $G_j$ by calculating the average distance between each node in one graph to its nearest node in the other graph, and normalizing this distance by the diameter of the full 50K point dataset. The result was $0.036$ average error.

\begin{table}[!htbp]
    \centering
    \begin{tabular}{llll}
        \hline
        \textbf{Method} & \textbf{Radius} & \textbf{$\#$ Simplices} & \textbf{Time (seconds)} \\
        \hline\hline
        {\bf Our Method}  & $\infty$ & 209397755 & 16089.4 \\
        \hline
        Baseline & .25 & 77261 & .15 \\
        \hline
        Baseline & .75 & 263787145 & 1485.42 \\
        \hline\hline

    \end{tabular}
    \caption{Image Patches dataset: Comparison of running time of \DMPCD{} and \mybaseline{}. Our algorithm has radius $\infty$ as we run on the full sparse DTM-Rips filtration. } 
    \label{tab:3circle_times}
\end{table}{}

\begin{figure}
\begin{center}
\includegraphics[width = 0.9\linewidth]{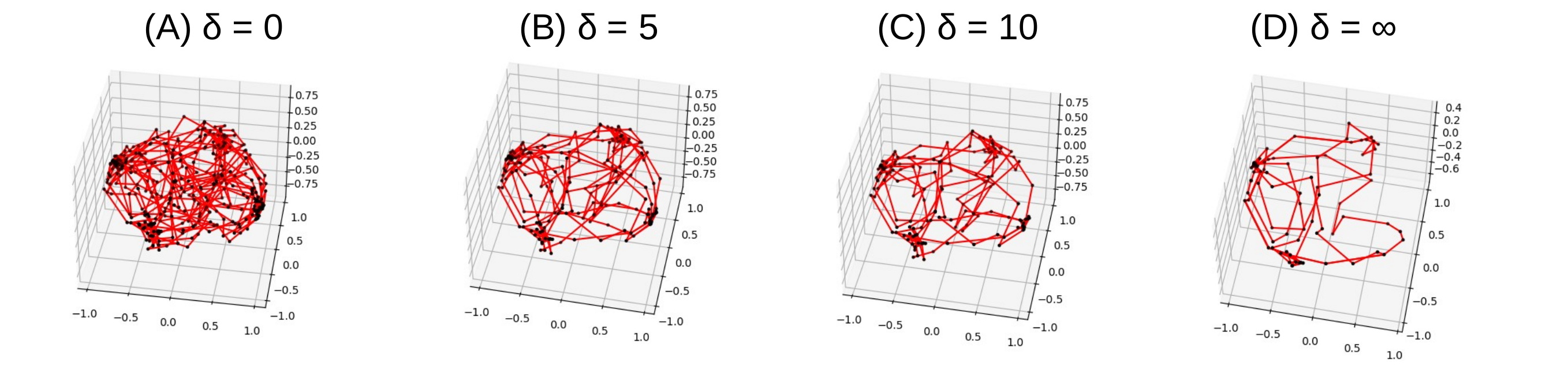}
\end{center}
\vspace*{-0.1in}\caption{{\small Image Patches: Outputs of \mybaseline{} with $\rips^r(P)$ complex ($r = .75$) as the input triangulation at various persistence thresholds. (A) $\delta = 0$ - With no thresholding of critical edges, the output captures the three circles that we expect to, as well as additional loops.  (B) $\delta = 5$ - Raising the persistence threshold allows for some of the additional loops to be removed while keeping the three circles we expect. (C) $\delta = 10$ - Further raising the persistence threshold results in losing part of the horizontal circle while keeping extra loops. (D) $\delta = \inf$ - Removing all critical edges except those with infinity persistence removes more of the desired three circle output and keeps loops not part of the desired output. }}
\label{fig:3circle_baseline}
\end{figure}

\begin{figure}
\begin{center}
\includegraphics[width = 0.9\linewidth]{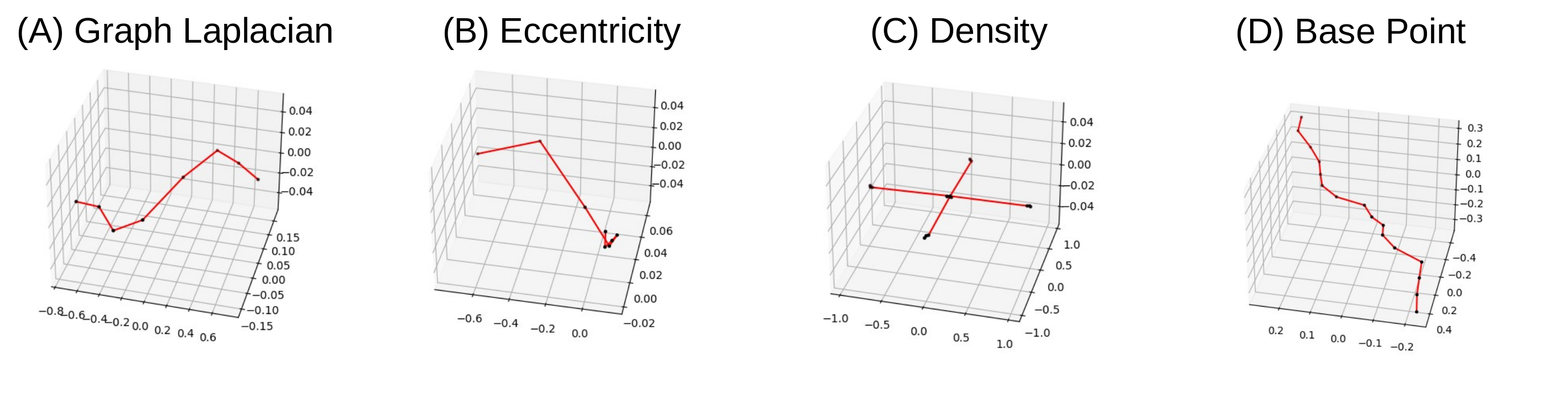}
\end{center}
\vspace*{-0.1in}\caption{{\small Image patches: Outputs of \myMapper{} using different filter functions - (A) graph Laplacian filter ($k = 15$), (B) eccentricity filter, (C) Gaussian density filter, (D), distance to base point filter. The dataset is too noisy and none of the filters result in \myMapper{} outputting a graph representative of the true underlying structure.}}
\label{fig:3circle_mapper}
\end{figure}

\myparagraph{Traffic flow dataset.} 
We also test on point clouds derived from traffic flow data \cite{caltrans}. We extract two datasets: the time-series of traffic flow at detector $\#$409529 from time-range 10/1/2017 to 10/14/2017 and from time-range 11/19/2017 to 12/2/2017 (which includes Thanksgiving). Each time-series is mapped to a point cloud dataset in $\reals^6$ via time-delay embedding.

Given a time series $f: t -> \reals$ and a parameter $\tau$, the lift defined by $\phi (t) = (f(t), f(t + \tau), ... , f(t + M \tau))$ is called a time delay embedding.  For each traffic flow function, we create a PCD using a time delay embedding with $M = 5$ and $\tau = 50$.  The two dimensional projections of these PCDs are shown in Figure 4 (B) of the main paper.  The first function's time delay embedding projection appears to be a single loop, while the second appears to have an inner loop and an outer loop.

The main paper shows the results of our \DMPCD{} algorithm with $k = 30$ on both time series datasets.  So far in our experiments, a default value of $k = 15$ has been used. By the nature of time delay embeddings, which may create clumps of points close together, different values of $k$ can produce markedly different results.  Shown in Figure \ref{fig:traffic_our_method} are results of \DMPCD{} on the two datasets with $k$ values of 15, 30, and 40.  For the first dataset (10/1/2017 - 10/14/2017), a single loop is captured with all values of $k$. For the second dataset (11/19/2017 - 12/2/2017), changing the value of $k$ results in more drastic changes in the output.  The persistence thresholds for the outputs are 8.25 ($k = 15$), 12.5 ($k = 30$), and 12.84 ($k = 40$).  In all cases, if the persistence threshold were raised enough to further threshold the output, a portion of the outer loop would be lost.  We note that our output must be connected, so the desired result is two loops with a single connection. The output with $k = 15$ contains many extra connections, while the output with $k = 30$ contains a single extra connection.  With $k = 40$, the desired output is achieved.

Also shown in the main paper, \mybaseline{} is able to successfully capture the single loop of the first time series dataset.  Results for \mybaseline{} on the second time series dataset are shown in Figure \ref{fig:traffic_baseline_k}.  The persistence thresholds for the outputs are 8 ($k = 15$), 5 ($k = 30$), and 3 ($k = 40$).  Just like the results for \DMPCD{} in Figure \ref{fig:traffic_our_method}, if the persistence thresholds were raised enough to further threshold the output, a portion of the outer loop would be lost, making the output of \DMPCD{} superior.

Running time and simplicial complex size comparisons for 10/1/2017 - 10/14/2017 and 11/19/2017 - 12/2/2017 traffic flows are shown in Table \ref{tab:traffic_times} and \ref{tab:turkey_times} respectively.  Note that the baseline results shown in the main paper use $r = 90$ and $r = 75$ respectively.  For the first dataset, we see that the number of simplices used by the \mybaseline{} with $r = 90$ is more than five times greater than the number of simplices used in \DMPCD{}. This results in longer running time for \mybaseline{}.  For the second dataset, the number of simplices used by the \mybaseline{} with $r = 75$ is a little less than three times greater than the number of simplices used in \DMPCD{}.  While the running time remained shorter for \mybaseline{} in this particular instance, we note that an increase in the value of $r$ can add a significant amount of simplices and push the running time to be longer than the \DMPCD{} running time.  We again note that for all values of $r$, the number of simplices used in \mybaseline{} would be the same number of simplices used by \myReeb{}.

Additionally, the main paper shows that \myMapper{} was able to extract the structure behind traffic flow from 10/1/2017 to 10/14/2017, but was unable to do so for traffic flow from 11/19/2017 to 12/2/2017.  Results of \myMapper{} on both datasets using a variety of filter functions is shown in Figure \ref{fig:traffic_mapper}.  Again, using the graph Laplacian and distance to base point filters allowed \myMapper{} to extract the single loop structure behind the first dataset.  However, \myMapper{} is unable to extract the two loop structure behind the second dataset with these filters, along with eccentricity and Gaussian density filters. In contrast, \DMPCD{} was able to get the true underlying structure behind both datasets.

\begin{table}[!htbp]
    \centering
    \begin{tabular}{llll}
        \hline
        \textbf{Method} & \textbf{Radius} & \textbf{$\#$ Simplices} & \textbf{Time (seconds)} \\
        \hline\hline
        {\bf Our Method}  & $\infty$ & 6,879,338 & 98.6 \\
        \hline
        Baseline & 75 & 14,646,522 & 54.7 \\
        \hline
        Baseline & 90 & 35,543,784 & 123.903 \\
        \hline\hline

    \end{tabular}
    \caption{Traffic (10/1/2017 - 10/14/2017) dataset: comparison of running time of \DMPCD{} and \mybaseline{}. Our algorithm has radius $\infty$ as we run on the full sparse DTM-Rips filtration. } 
    \label{tab:traffic_times}
\end{table}{}

\begin{table}[!htbp]
    \centering
    \begin{tabular}{llll}
        \hline
        \textbf{Method} & \textbf{Radius} & \textbf{$\#$ Simplices} & \textbf{Time (seconds)} \\
        \hline\hline
        {\bf Our Method}  & $\infty$ & 5,720,309 & 106.1 \\
        \hline
        Baseline & 60 & 5,157,336 & 16.8 \\
        \hline
        Baseline & 75 & 15,659,797 & 57.7 \\
        \hline\hline

    \end{tabular}
    \caption{Traffic (11/19/2017 - 12/2/2017) dataset: comparison of running time of \DMPCD{} and \mybaseline{}. Our algorithm has radius $\infty$ as we run on the full sparse DTM-Rips filtration. } 
    \label{tab:turkey_times}
\end{table}{}

\begin{figure}
\begin{center}
\includegraphics[width = 0.9\linewidth]{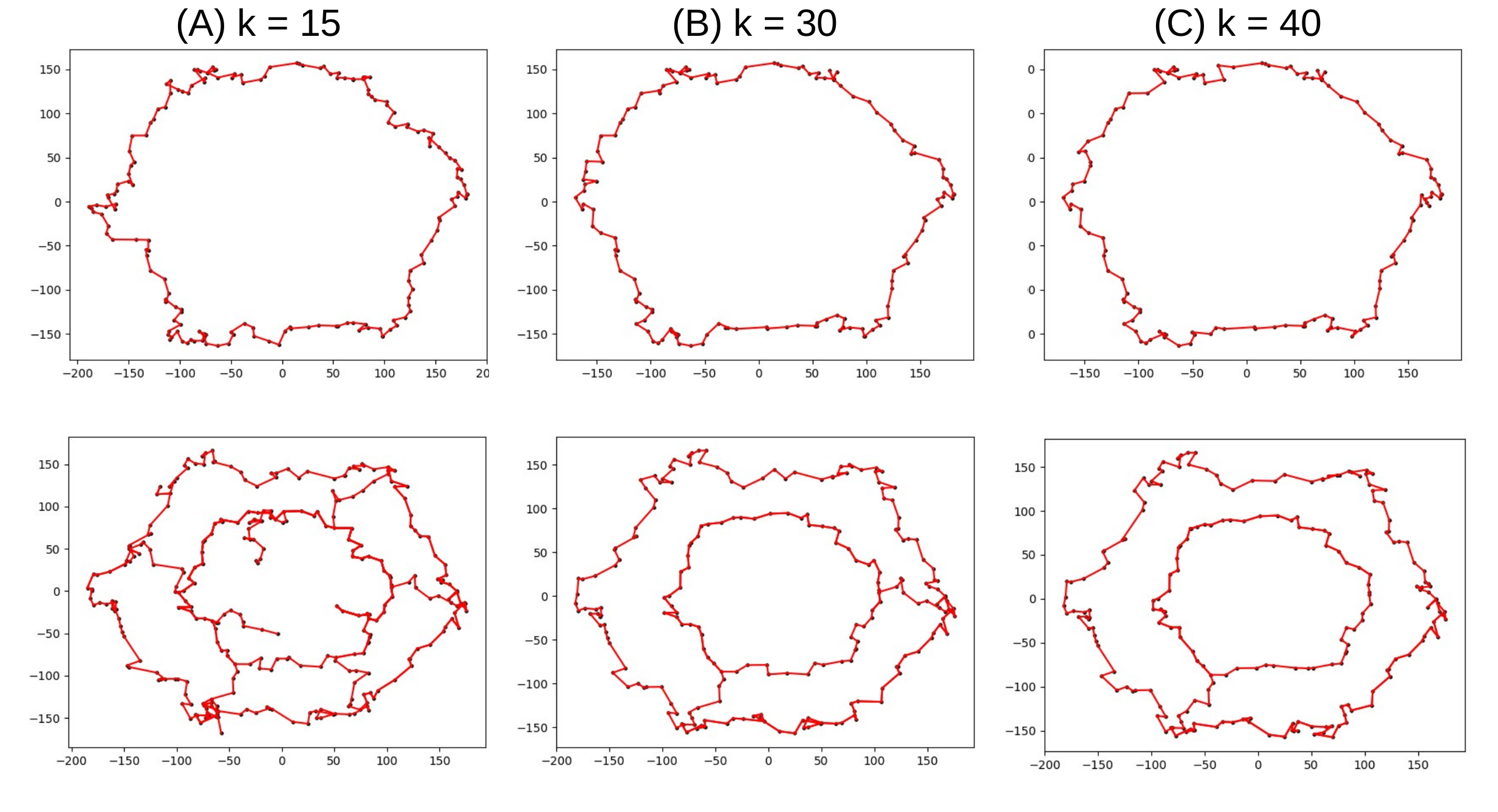}
\end{center}
\vspace*{-0.1in}\caption{{\small Traffic flow: Outputs of \DMPCD{} on both datasets (10/1/2017 - 10/14/2017 top row, 11/19/2017 - 12/2/2017 bottom row) with different values of $k$. (A) $k = 15$ - Single loop captured in first dataset, two loops captured with extra connections in second dataset.  (B) $k = 30$ - Single loop captured in first dataset, two loops captured with an extra connection in second dataset. (C) $k = 40$ - Single loop captured in first dataset, two loops captured with a single connection in second dataset.}}
\label{fig:traffic_our_method}
\end{figure}

\begin{figure}
\begin{center}
\includegraphics[width = 0.9\linewidth]{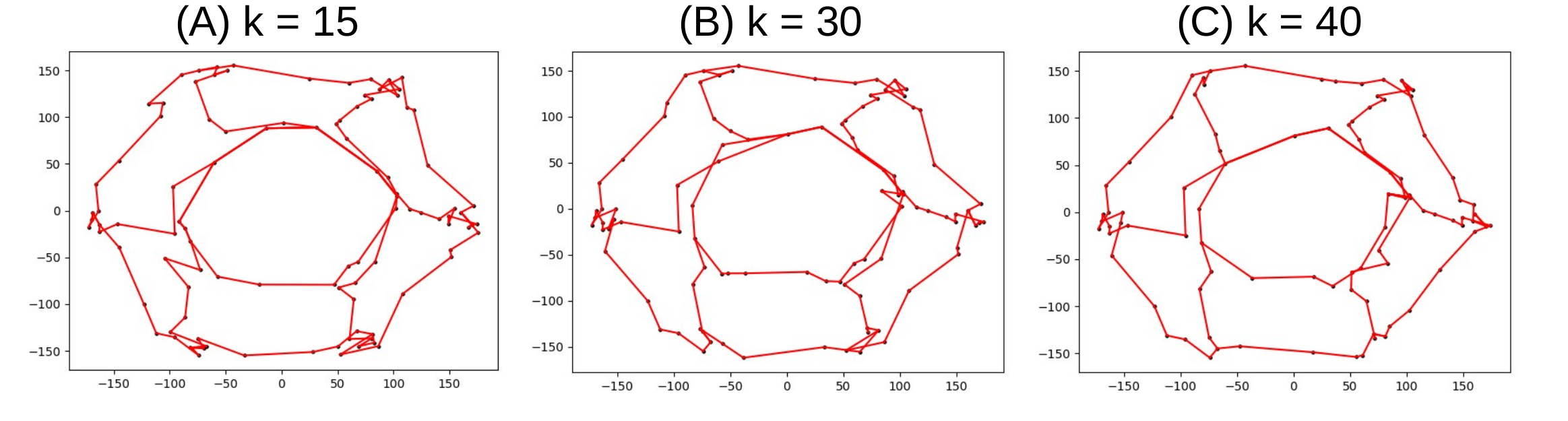}
\end{center}
\vspace*{-0.1in}\caption{{\small Traffic flow: Outputs of \mybaseline{} method on the second traffic time series dataset (11/19/2017 - 12/2/2017) at different values of $k$ - (A) $k = 15$, (B) $k = 30$, (C) $k = 40$. In all cases, any further simplification will lose the outer loop before removing any connections with the inner loop.}}
\label{fig:traffic_baseline_k}
\end{figure}

\begin{figure}[H]
\begin{center}
\includegraphics[width = 0.9\linewidth]{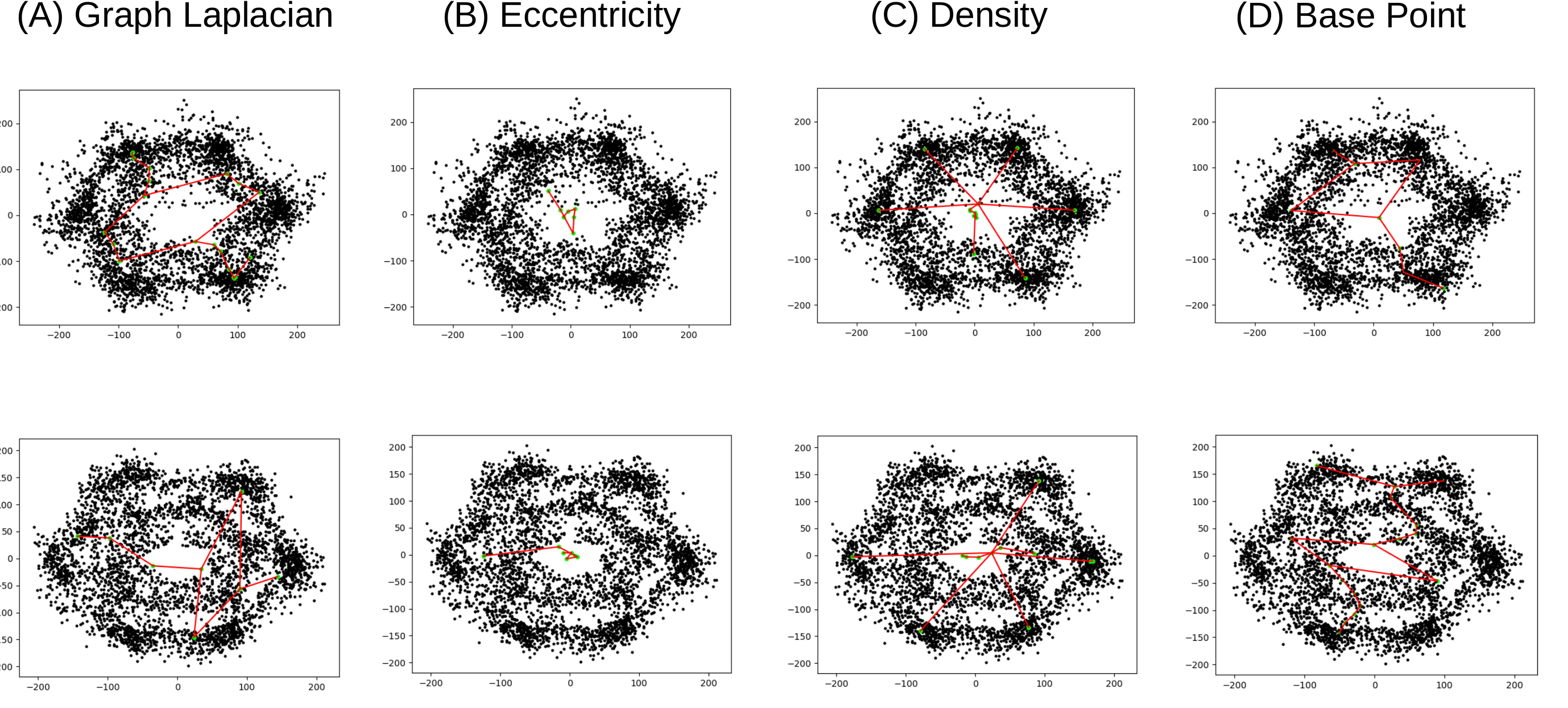}
\end{center}
\vspace*{-0.1in}\caption{{\small Traffic flow: Outputs of \myMapper{} on both datasets using different filter functions - (A) graph Laplacian filter ($k = 15$), (B) eccentricity filter, (C) Gaussian density filter, (D), distance to base point filter. The structure of traffic flow from 10/1/2017 to 10/14/2017 (first row) is captured by \myMapper{} using the graph Laplacian filter function and the distance to base point filter function, while the structure of traffic flow from 11/19/2017 to 12/2/2017 (second row) is not captured by \myMapper{} using any of the filter functions. }}
\label{fig:traffic_mapper}
\end{figure}

\myparagraph{Coil-20.}
Similarly to the two circle example, both the \mybaseline{} and \myReeb{} approaches are unable to capture all coils because the coils have varying scales in the original space.  A concrete example is shown in Figure \ref{fig:coil-baseline}, where Objects 1 and 17 cannot be captured at the same scale.
We also applied \myMapper{} to Coil-17 using a base point filter.  While \myMapper{} can extract the structure of individual objects quite well, the method also struggles to capture all coils.  Focusing on Objects 1 and 17 again, we see that by changing the epsilon parameter of the density clustering scheme we use inside of \myMapper{}, we are able to capture the structure of either Object 1 or Object 17, but not both (results shown in Figure \ref{fig:coil-mapper}).  While it may be possible that a different clustering scheme (or different covering and filter function combinations) could lead to a \myMapper{} configuration that can capture both objects, finding parameters able to capture all coils would be difficult.

\begin{figure}[H]
\begin{center}
\includegraphics[width = 0.9\linewidth]{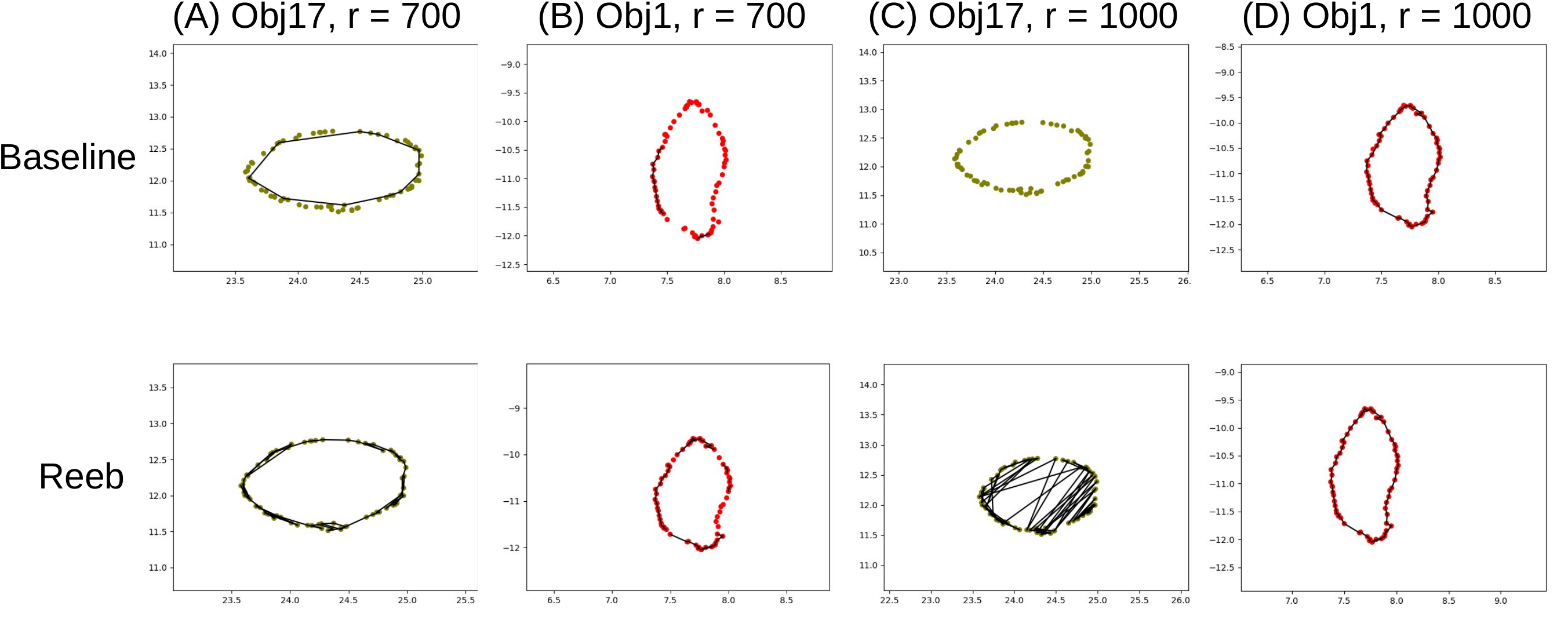}
\end{center}
\vspace*{-0.1in}\caption{{\small Coil: Objects 1 and 17 with \mybaseline{} and \myReeb{} outputs.  (A) Object 17, $r = 700$ - Object 17 is captured by both methods.  (B) Object 1, $r = 700$ - Object 1 is not captured by either method.  (C) Object 17, $r = 1000$ - Object 17 is not captured by either method. (D) Object 1, $r =1000$ - Object 1 is captured by both methods. }}
\label{fig:coil-baseline}
\end{figure}

\begin{figure}[H]
\begin{center}
\includegraphics[width = 0.9\linewidth]{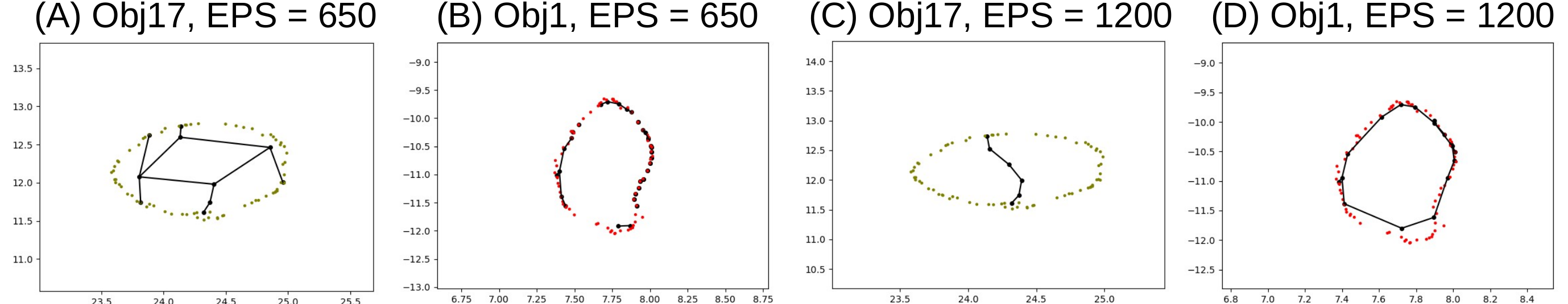}
\end{center}
\vspace*{-0.1in}\caption{{\small Coil: Objects 1 and 17 with \myMapper{} outputs generated using a base point filter function.  (A) Object 17, $EPS = 650$ - structure of Object 17 is captured.  (B) Object 1, $EPS = 650$ - structure of Object 1 is not captured.  (C) Object 17, $EPS = 1200$ - structure of Object 17 is not captured. (D) Object 17, $EPS = 1200$ - structure of Object 1 is captured. }}
\label{fig:coil-mapper}
\end{figure}

\end{document}